\documentclass[10pt,journal,compsoc,twocolumn]{IEEEtran}
\pdfoutput=1

\usepackage[nocompress]{cite}
\usepackage{threeparttable}
\usepackage{booktabs}
\usepackage{float}
\usepackage{extarrows}
\usepackage{graphicx}
\usepackage{epstopdf}
\usepackage{xcolor}
\usepackage{epsfig}

\usepackage{amsmath}
\usepackage{amssymb}
\usepackage{amsthm}
\usepackage{color}
\usepackage{url}
\usepackage{array}
\usepackage{multirow}
\usepackage{algorithm}
\usepackage{algorithmic}
\usepackage{bigstrut}
\usepackage[T1]{fontenc}
\usepackage[justification=centering]{caption}

\usepackage{rotating}

\newtheorem{definition}{Definition}
\newtheorem{theorem}{Theorem}
\newtheorem{lemma}{Lemma}
\newtheorem{corollary}{Corollary}
\newtheorem{proposition}{Proposition}

\numberwithin{equation}{section}
\numberwithin{definition}{section}
\numberwithin{theorem}{section}
\numberwithin{lemma}{section}
\numberwithin{corollary}{section}
\numberwithin{proposition}{section}
\numberwithin{remark}{section}
\numberwithin{assumption}{section}
\newlength{\figurewidth}
\allowdisplaybreaks[4]

\begin{document}

\title{Robust Matrix Completion with Deterministic Sampling via Convex Optimization}

\author{Yinjian~Wang
\IEEEcompsocitemizethanks{\IEEEcompsocthanksitem Y. Wang is with the School of Information and Electronics,
	Beijing Institute of Technology, and Beijing Key Laboratory of Fractional
	Signals and Systems, 100081 Beijing, China, and also with the Institute of Methodologies for Environmental Analysis, National Research Council (CNR-IMAA), 85050 Tito, Italy (e-mail: yinjw@bit.edu.cn).}}

{}

\IEEEcompsoctitleabstractindextext{%

\begin{abstract}
This paper deals with the problem of robust matrix completion---retrieving a low-rank matrix and a sparse matrix from the compressed counterpart of their superposition. Though seemingly not an unresolved issue, we point out that the compressed matrix in our case is sampled in a deterministic pattern instead of those random ones on which existing studies depend. In fact, deterministic sampling is much more hardware-friendly than random ones. The limited resources on many platforms leave deterministic sampling the only choice to sense a matrix, resulting in the significance of investigating robust matrix completion with deterministic pattern. In such spirit, this paper proposes \textit{restricted approximate $\infty$-isometry property} and proves that, if a \textit{low-rank} and \textit{incoherent} square matrix and certain deterministic sampling pattern satisfy such property and two existing conditions called  \textit{isomerism} and \textit{relative well-conditionedness}, the exact recovery from its sampled counterpart grossly corrupted by a small fraction of outliers via convex optimization happens with very high probability. 
\end{abstract}
\begin{IEEEkeywords}
Matrix completion, compressed robust principle component analysis, deterministic sampling, identifiability, low-rankness, sparse.
\end{IEEEkeywords}}
\maketitle
\IEEEdisplaynotcompsoctitleabstractindextext
\IEEEpeerreviewmaketitle

\IEEEraisesectionheading{\section{Introduction}\label{sec:introduction}}
\IEEEPARstart{T}{he} presence of missing data is ubiquitous in real world. For example, in the famous \textsl{Netflix Prize} \cite{bennett2007kdd} problem, people are looking for computational methods to infer any user's preference on any movie, with only a few ratings per user given. Applications like this stimulated the researches on the problem of Matrix Completion (MC) \cite{johnson1990matrix} which concerns the recovery of an unknown matrix from a fraction of its entries. Though sounds mysterious, Cand\`es \textit{et al.} \cite{candes2010power,WOS:000272299900003} pioneered to unveil that exact recovery of an $n\times n$ matrix with rank $r$ is possible at a sampling rate of order $n^ar^b\log^c(n)$. The conditions they established, \textit{i.e.,} \textit{low-rankness} and \textit{incoherence}, has become the core of the subsequent many MC researches, including those improving the theoretical sampling boundaries \cite{keshavan2010matrix,recht2011simpler,7064749} and those developing efficient algorithms \cite{doi:10.1137/080738970,6262492,6979050}.

Yet, those previous theories or algorithms lay on the assumption that the sampling mechanism of those observed entries are uniformly random. Such an assumption may seem in line with the cases where the data missing happens due to non-human factors, \textit{e.g.,} the recommendation systems \cite{8400447}. But there are a lot of other cases where the sampling is performed by man-made systems. The complexity of implementing a random sampling hinders its application from tasks requiring speed with limited hardware resources. In such task, \textit{e.g.,} snapshot compressive imaging \cite{9363502}, the sampling pattern has to be fixed and deterministic, promoting the studies on MC with Deterministic Sampling (MCDS). Chen \textit{et al.} \cite{chen2014coherent} studied MC with the observed entries sampled proportionally to the local row and column coherences. In \cite{bhattacharya2022matrix}, the probability regarding each entry being available is supposed as a function of the entry itself. Besides, Pimentel-Alarcón \textit{et al.} \cite{pimentel2016characterization} characterized the conditions on deterministic sampling for finite completability. Shapiro \textit{et al.} \cite{shapiro2018matrix} established well-posedness condition for the local uniqueness of minimum rank matrix completion solutions. Based on graph limit theory, Chatterjee \textit{et al.} \cite{chatterjee2020deterministic} provided asymptotically solvable guarantee for sequential MC problems with arbitrary sampling patterns. Moreover, Burnwal \textit{et al.} \cite{burnwal2020deterministic} derived sufficient recovery conditions for the sampling set chosen as the edge set of an asymmetric Ramanujan bigraph. However, these literature either studied the sampling patterns that are not totally arbitrary/deterministic, or presented recovery theories that are only asymptotic/approximate. In spirit of searching the global exact missing data recovery guarantee for arbitrary deterministic sampling, Liu \textit{et al.} \cite{liu2019matrix} proposed \textit{isomerism} \cite{liu2017new} and \textit{relative well-conditionedness}, using which they concluded a series of theories for MCDS based on different models including convex optimization. Later on, for those sampling violating \textit{isomerism}, they further proposed \cite{9851461} convolutional low-rankness and Convolutional Nuclear Norm (CNN) to realize MCDS with arbitrary sampling via convex optimization.

Another noteworthy issue in MC is its robustness to noises/outliers. In fact, even a small fraction of noised entries can severely jeopardize the low-rankness of the complete matrix \cite{10.1109/TIT.2011.2173156}. Since being low-dimensionally structured is the very basic idea behind any high-dimensional data reconstruction tasks \cite{Wright_Ma_2022}, it is of significance to study Robust MC (RMC). Cand\`es \textit{et al.} \cite{candes2010matrix} demonstrated the possibility of approximate matrix completion from noisy sampled entries. Keshavan \textit{et al.} \cite{keshavan2009matrix} generalized \textsc{OptSpace} \cite{keshavan2010matrix} to the noisy case and proved performance guarantees that are order-optimal in a number of circumstances. Furthermore, as generalization to Robust Principle Component Analysis (RPCA), Cand{\`e}s \cite{candes2011robust} revealed that RMC can be exactly solved through very convenient convex optimization. The theoretical recovery guarantee therein was later improved to a result with fewer log factors \cite{li2013compressed}. Besides, Cherapanamjeri \cite{cherapanamjeri2017nearly} presented projected gradient descent-based algorithm that solves RMC using nearly optimal number of observations while tolerating a nearly optimal number of corruptions. Nevertheless, these theories all relent to the setting of random sampling. There has been very few studies concerning RMC with Deterministic Sampling (RMCDS). The results in \cite{klopp2017robust} apply to deterministic sampling but only approximate guarantee with error bounds instead of exact recovery theory was given. Finite and unique completablity was discussed in \cite{ashraphijuo2017deterministic}, but there existed no explicit way to actually reach the expected solution. In a nutshell, the exact recovery of RMCDS remains a limit to be broken through.

In response, this paper presents, to the best of our knowledge, the very first theory guaranteeing that RMCDS problem is exactly solvable subject to arbitrary deterministic sampling. Specifically, we propose \textit{restricted approximate $\infty$-isometry property}. It is then established that any square matrix being \textit{low-rank} and \textit{incoherent} with a sampling pattern satisfying this property, \textit{isomerism} and \textit{relative well-conditionedness} can be exactly recovered via convex optimization from its corrupted sampled entries with probability at least $1-n^{-C_m}$, provided that the corruption happens with a uniformly small probability.

The remainder of this paper is organized as follows. Section \ref{sec:Notations} introduces the notation system of this paper. Section \ref{sec:formulation} presents the convex formulation of the RMCAS problem, and more importantly, discusses the model assumptions on which our recovery theory lies. Section \ref{sec:Identifiability} delivers our main theorem and provides an outline of its proof. Detailed mathematical deduction are contained in Section \ref{sec:Mathematical Proofs} while Section \ref{sec:Conclusion} concludes the paper.

\section{Notations}
\label{sec:Notations}
In this paper, non-boldface letters are used to denote scalars. Boldface upper and lowercase letters denote matrices and vectors, respectively. Specifically, we use $\boldsymbol{e}_i$ to denote unit vector where the $i\mbox{-}$th element is $1$ with others being zeroes and similarly, $\boldsymbol{E}_{ij}$ to denote matrix where the $ij\mbox{-}$th element is $1$ while others are zeroes. $D_{ij}$ means the $ij\mbox{-}$th element of the matrix $\boldsymbol{D}$. Besides, sans-serif letters like $\mathsf{O},\mathsf{W},\mathsf{V}$ denote sets, with $[n]\triangleq\left\{1,2,\cdots,n\right\}.$ Suppose $\mathsf{O}\in [n]\times[n]$ and its elements are independently selected from $[n]\times[n]$ subject to Bernoulli model with probability $\rho$, we will denote it with $\mathsf{O}\thicksim Ber(\rho).$ Moreover, with slight notation abuse, calligraphic letters denote linear spaces or linear operators. Conforming to this, we reserve $\mathcal{P}$ in particular to denote orthogonal projectors, and $\mathcal{P}_\mathcal{T}$ represents the orthogonal projection onto some linear space $\mathcal{T}$. $\mathcal{T}^\bot$ $\mathsf{O}^c$ are the orthogonal complement and complement of $\mathcal{T}$ and $\mathsf{O}$, respectively. $\mathcal{T}+\mathcal{W}$ gives the direct sum of two linear spaces.  $\mathcal{E}$ denotes expectation operator and $\mathcal{I}$ stands for identity transformation. $Pr\left(\cdot\right)$ returns the probability of certain event. $\displaystyle sgn(x)=\left\{\begin{array}{l}
	x/|x|,\,x\neq0\\
	0,\,x=0
\end{array}\right.$ denotes the symbolic function and is element-wise when imposed on matrix.

We use $\left\langle\cdot,\cdot\right\rangle$ to denote matrix inner product. $\boldsymbol{D}^*,\,\mathcal{P}^*$ are the Hermitian transposes of $\boldsymbol{D},\,\mathcal{P}$, respectively. $\left\|\cdot\right\|_2$ is the $l_2\mbox{-}$norm of a vector. $\left\|\cdot\right\|$ and $\left\|\cdot\right\|_F$ denote the operator norm and Frobenius norm of a matrix, respectively. Note that the operator norm of an operator $\mathcal{P}:\mathbb{R}^{M\times N}\rightarrow\mathbb{R}^{M'\times N'}$ is also denoted with $\left\|\cdot\right\|$, whch is defined as $\left\|\mathcal{P}\right\|=\sup_{\boldsymbol{D}}\,\frac{\left\lVert\mathcal{P}\left[\boldsymbol{D}\right]\right\lVert_{F}}{\left\lVert\boldsymbol{D}\right\lVert_{F}}$. $\left\|\boldsymbol{D}\right\|_\infty\triangleq \max_{ij}\,\mid D_{ij}\mid$ is defined as the $l_\infty\mbox{-}$norm of a matrix. $\left\|\cdot\right\|_*$ and $\left\|\cdot\right\|_1$ represents the nuclear norm and $l_1\mbox{-}$ norm, respectively.

\section{Problem Formulation}
\label{sec:formulation}
Suppose we have a low-rank matrix $\boldsymbol{L}_0\in\mathbb{R}^{n\times n}$ and a sparse matrix $\boldsymbol{S}_0\in\mathbb{R}^{n\times n}$. Denote $\boldsymbol{Y}=\boldsymbol{L}_0+\boldsymbol{S}_0$ as their superposition. In this paper we seek to identify $\boldsymbol{L}_0$ from partial observation of $\boldsymbol{Y}$, \textit{i.e.,} $\mathcal{P}_{\mathcal{O}}\left[\boldsymbol{Y}\right]$, where $\mathcal{O}$ denotes the subspace of matrices supported on some deterministic set $\mathsf{O}\subseteq [n]\times[n].$\footnote{Note that this will simultaneously identify $\mathcal{P}_{\mathcal{O}}\left[\boldsymbol{S}_0\right]$.} Due to the tractability and scalability that convex optimization exhibited in related problems (\textit{e.g.,} \cite{candes2010power,WOS:000272299900003,recht2011simpler,liu2019matrix,9851461,candes2011robust,li2013compressed}), we plan to formulate our RMCDS problem into convex optimization again. That is to say, the following optimization problem is to be discussed:
\begin{equation}
	\begin{aligned}
		&\min_{\boldsymbol{L},\boldsymbol{S}}\,\left\lVert\boldsymbol{L}\right\lVert_*+\lambda\left\lVert\boldsymbol{S}\right\lVert_1,\\
		s.t.&,\, \mathcal{P}_{\mathcal{O}}\left[\boldsymbol{Y}\right]=\mathcal{P}_{\mathcal{O}}\left[\boldsymbol{L}+\boldsymbol{S}\right].
	\end{aligned}
	\label{eq: Integrated Formulation}
	\nonumber
\end{equation}
To discuss its exact recovery conditions, below we introduce some concepts and premises about the low-rank matrix $\boldsymbol{L}_0$, sparse matrix $\boldsymbol{S}_0$ and the sampling set $\mathcal{O}$. For the sake of brevity, these premises will be taken for default without being further mentioned in the remaining sections, unless needed. 

\textbf{Tangent space at $\boldsymbol{L}_0$:} Suppose that $\boldsymbol{L}_0$ is with rank $r$ and compact Singular Value Decomposition (SVD) $\boldsymbol{L}_0=\boldsymbol{U}\boldsymbol{\Pi}\boldsymbol{V}^\ast$. Define $\mathcal{T}\triangleq\left\{\boldsymbol{U}\boldsymbol{R}^\ast+\boldsymbol{Q}\boldsymbol{V}^\ast\vert\boldsymbol{R},\boldsymbol{Q}\in\mathbb{R}^{n\times r}\right\}$ the tangent space to the set of rank-$r$ matrices at $\boldsymbol{L}_0$.

\textbf{Incoherence of $\boldsymbol{L}_0$:} There exists a constant $\nu\in[1,\frac{n}{r}],$ such that
$$\max\left\{\max_{i\in[n]}\left\lVert\boldsymbol{U}^*\boldsymbol{e}_i\right\lVert_2^2,\,\max_{i\in[n]}\left\lVert\boldsymbol{V}^*\boldsymbol{e}_i\right\lVert_2^2\right\}\leq \frac{\nu r}{n}$$ and $$\left\lVert\boldsymbol{U}\boldsymbol{V}^*\right\lVert_\infty\leq \frac{\sqrt{\nu r}}{n\sqrt{\log(n)}}.$$
Note that the second inequality is strengthened from $\left\lVert\boldsymbol{U}\boldsymbol{V}^*\right\lVert_\infty\leq \frac{\sqrt{\nu r}}{n}$ in \cite{candes2011robust,li2013compressed} for the more challenging deterministic case.

\textbf{Generation mechanism of $\boldsymbol{S}_0$:} Different from man-made sampling, the noise component $\boldsymbol{S}_0$ comes from mainly unpredictable non-human factors. Thus it is appropriate assuming its generation mechanism to be uniformly random, as in \cite{candes2011robust,li2013compressed,Wright_Ma_2022}. Specifically, we suppose that $\boldsymbol{S}_0$ is supported on $\mathsf{W}\sim Ber(\rho)\subseteq [n]\times[n]$ and denote that  $\overline{\boldsymbol{S}}_0=\mathcal{P}_{\mathcal{O}}\left[\boldsymbol{S}_0\right],\,\overline{\boldsymbol{\Sigma}}_0=sgn\left(\overline{\boldsymbol{S}}_0\right),$ while the non-zero entries in $\overline{\boldsymbol{\Sigma}}_0$ are assumed as Rademacher $\pm1$ random variables. Additionally, for convenience we denote $\mathsf{V}=\mathsf{O}\cap\mathsf{W}$ on which those observed corrupted entries are supported and $\mathsf{N}=\mathsf{O}\cap\mathsf{W}^\bot$ on which those observed clean entries settle. $\mathcal{V}$ and $\mathcal{N}$ are then denoted as the linear subspaces in which the matrices are supported on $\mathsf{V}$ and $\mathsf{N}$, respectively.

\textbf{Generation mechanism of $\boldsymbol{S}_0$:} Different from man-made sampling, the noise component $\boldsymbol{S}_0$ comes from mainly unpredictable non-human factors. Thus it is appropriate assuming its generation mechanism to be uniformly random, as in \cite{candes2011robust,li2013compressed,Wright_Ma_2022}. Specifically, we suppose that $\boldsymbol{S}_0$ is supported on $\mathsf{W}\sim Ber(\rho)\subseteq [n]\times[n]$ and denote that  $\overline{\boldsymbol{S}}_0=\mathcal{P}_{\mathcal{O}}\left[\boldsymbol{S}_0\right],\,\overline{\boldsymbol{\Sigma}}_0=sgn\left(\overline{\boldsymbol{S}}_0\right),$ while the non-zero entries in $\overline{\boldsymbol{\Sigma}}_0$ are assumed as Rademacher $\pm1$ random variables. Additionally, for convenience we denote $\mathsf{V}=\mathsf{O}\cap\mathsf{W}$ on which those observed corrupted entries are supported and $\mathsf{N}=\mathsf{O}\cap\mathsf{W}^\bot$ on which those observed clean entries settle. $\mathcal{V}$ and $\mathcal{N}$ are then denoted as the linear subspaces in which the matrices are supported on $\mathsf{V}$ and $\mathsf{N}$, respectively.

\textbf{Isomerism and relative well-conditionedness:} In regard to the sampling pattern, unlike the random-sampling case in which the establishment of identifiabiliy only involves the sampling rate, an identifiable deterministic sampling pattern is $\boldsymbol{L}_0$-dependent. That is, the 'ideal' sampling pattern varies with the desired $\boldsymbol{L}_0$. Thus, we require the sampling pattern to enable $\boldsymbol{L}_0$ to be $\mathsf{O}/\mathsf{O}^T$-isomeric with $\gamma_{\mathsf{O},\mathsf{O}^T}\left(\boldsymbol{L}_0\right)>\frac{3}{4}$, as required by \cite{liu2019matrix}. However, it is not that random sampling is totally independent of the desired matrix. On the contrary, we don't need to additionally suppose the mutual dependence between $\boldsymbol{L}_0$ and random sampling because such dependence comes from the nature of random sampling, as shown by Thm 3.1 and 3.3 in \cite{liu2019matrix}.
\section{Exact Identifiability Theories}
\label{sec:Identifiability}
This section delivers our main theorem regarding the exact recovery of the RMCDS problem via convex optimization formulated by the conventional nuclear norm and $l_1$-norm. Concretely, on what conditions does the solution $\left(\boldsymbol{L}_\star,\boldsymbol{S}_\star\right)$ to
\begin{equation}
	\begin{aligned}
		&\min_{\boldsymbol{L},\boldsymbol{S}}\,\left\lVert\boldsymbol{L}\right\lVert_*+\lambda\left\lVert\boldsymbol{S}\right\lVert_1,\\
		s.t.&,\, \mathcal{P}_{\mathcal{O}}\left[\boldsymbol{Y}\right]=\mathcal{P}_{\mathcal{O}}\left[\boldsymbol{L}+\boldsymbol{S}\right].
	\end{aligned}
	\label{eq: NN Formulation}
\end{equation}
satisfy
$$\boldsymbol{L}_\star=\boldsymbol{L}_0,\,\boldsymbol{S}_\star=\overline{\boldsymbol{S}}_0.\footnote{It is impossible to recover the whole $\boldsymbol{S}_0$ in any sense. So we only desire the recovery of $\overline{\boldsymbol{S}}_0$. This is in accordance with any related researches.}$$
\subsection{Restricted Approximate $\infty$-Isometry Property}
To begin with, we propose a new concept called Restricted Approximate $\infty$-Isometry Property (RAIIP, the first I stands for Infinity), which starts with the restricted $\infty$-norm.
\begin{definition}[Restricted $\infty$-norm of linear operators]
	\textit{Given a linear transformation $\mathcal{P}:\,\mathbb{R}^{n\times n}\rightarrow\mathbb{R}^{n\times n}$, its restricted $\infty$-norm, denoted by $\left\lVert\mathcal{P}\right\lVert_{\mathcal{T},\infty}$, is induced from the matrix $l_\infty$-norm while restricted on some linear space $\mathcal{T}$, written as
		$$\left\lVert\mathcal{P}\right\lVert_{\mathcal{T},\infty}=\sup_{\boldsymbol{D}\in\mathcal{T}\subseteq\mathbb{R}^{n\times n}}\,\frac{\left\lVert\mathcal{P}\left[\boldsymbol{D}\right]\right\lVert_{\infty}}{\left\lVert\boldsymbol{D}\right\lVert_{\infty}}.$$}
	\label{Def: Restricted infty-norm}
\end{definition}
\noindent
Similar to the restriction enforced by Thm 4.1 in \cite{WOS:000272299900003} for MC problem, imposing a bound on the restricted $\infty$-norm of $\mathcal{P}_{\mathcal{T}}\mathcal{P}_{\mathcal{O}^\bot}$, we obatin RAIIP.
\begin{definition}[Restricted Approximate $\infty$-Isometry Property]
	\textit{On the premises in Section \ref{sec:formulation}, if it holds that
		$$\left\lVert\mathcal{P}_{\mathcal{T}}\mathcal{P}_{\mathcal{O}^\bot}\right\lVert_{\mathcal{T},\infty}<1,$$ then we will say that $\mathcal{P}_{\mathcal{O}}$ satisfies RAIIP.}
	\label{Def: RAIIP}
\end{definition}
\noindent
This condition implies that $\left\lVert\mathcal{P}_{\mathcal{T}}-\mathcal{P}_{\mathcal{T}}\mathcal{P}_{\mathcal{O}}\right\lVert_{\mathcal{T},\infty}$ is not too large, which happens if $\mathcal{P}_{\mathcal{O}}$ 'approximately' preserves the length (in terms of $l_\infty$-norm) of any elements in $\mathcal{T}$. Namely, $\mathcal{P}_{\mathcal{O}}$ is 'approximately' isometric restricted on $\mathcal{T}$. 
\subsection{Main Theorem}
With RAIIP at hand, we now present our main identifiability theorem.
\begin{theorem}
	\textit{With the premises in Section \ref{sec:formulation}, if $\mathcal{P}_{\mathcal{O}}$ satisfies RAIIP, additionally if  $\rho<\min\left\{\rho_s,1-C\frac{\nu r\log^2(n)}{n}\right\}$ with $C$ some positive numerical constant, then there exists another constant $C_m>0$, such that $\left(\boldsymbol{L}_0,\overline{\boldsymbol{S}}_0\right)$ is the unique solution to problem \eqref{eq: NN Formulation} with probability at least $1-n^{-C_m}$, provided that $\lambda=\frac{1}{\sqrt{n\log(n)}}$ and $\rho_s$ is small enough.}
	\label{Th: Main Theorem}
\end{theorem}
\subsection{Proof Architecture of the Main Theorem}
In the sequel, the proof of Thm \ref{Th: Main Theorem} is delivered in several steps. Firstly, by simply inspecting the KKT conditions of \eqref{eq: NN Formulation}, it is concluded that, if there exists a dual multiplier $\boldsymbol{\Lambda}\in\mathcal{O}$ such that
\begin{equation}
	\begin{aligned}
		&\mathcal{P}_{\mathcal{T}}\left[\boldsymbol{\Lambda}\right]=\boldsymbol{U}\boldsymbol{V}^\ast,\,&\left\lVert\mathcal{P}_{\mathcal{T}^\bot}\left[\boldsymbol{\Lambda}\right]\right\lVert< 1,
		\\
		&\mathcal{P}_{\mathcal{V}}\left[\boldsymbol{\Lambda}\right]=\lambda\overline{\boldsymbol{\Sigma}}_0,\,&\left\lVert\mathcal{P}_{\mathcal{N}}\left[\boldsymbol{\Lambda}\right]\right\lVert_\infty<\lambda.
	\end{aligned}
	\label{eq:Strict KKT}
\end{equation}
Then problem \eqref{eq: NN Formulation} is uniquely solvable conditioned on $\left\lVert\mathcal{P}_{\mathcal{N}^\bot}\mathcal{P}_{\mathcal{T}}\right\lVert<1.$ Unfortunately, the equalities in \eqref{eq:Strict KKT} are too strict for one to establish a qualified $\boldsymbol{\Lambda}.$ Thus, \eqref{eq:Strict KKT} is relaxed into \eqref{eq:Relaxed KKT} in Lemma \ref{Le: Relaxed KKT}.
\begin{lemma}
	\textit{Subject to the settings of Thm \ref{Th: Main Theorem}, $\left(\boldsymbol{L}_0,\overline{\boldsymbol{S}}_0\right)$ is the unique solution to problem \eqref{eq: NN Formulation} if there exists $\boldsymbol{\Lambda}\in\mathcal{O}$ such that
		\begin{equation}
			\begin{aligned}
				&\left\lVert\mathcal{P}_{\mathcal{T}}\left[\boldsymbol{\Lambda}\right]-\boldsymbol{U}\boldsymbol{V}^\ast\right\lVert_F< \frac{1}{n^2},\,&\left\lVert\mathcal{P}_{\mathcal{T}^\bot}\left[\boldsymbol{\Lambda}\right]\right\lVert< \frac{1}{2},
				\\
				&\mathcal{P}_{\mathcal{V}}\left[\boldsymbol{\Lambda}\right]=\lambda\overline{\boldsymbol{\Sigma}}_0,\,&\left\lVert\mathcal{P}_{\mathcal{N}}\left[\boldsymbol{\Lambda}\right]\right\lVert_\infty< \frac{\lambda}{2}.
			\end{aligned}
			\label{eq:Relaxed KKT}
	\end{equation}}
	\label{Le: Relaxed KKT}
\end{lemma}
\begin{proof}
	\textit{Deferred to Section \ref{sec:Mathematical Proofs}.}
\end{proof}
\noindent
Thereupon, the problem reduces to the construction of the dual certificate $\boldsymbol{\Lambda}$ satisfying \eqref{eq:Relaxed KKT}, which is answered in steps by Lemma \ref{Le: Certificate1} and \ref{Le: Certificate2}. In \cite{candes2011robust}, similar certificates are constructed via golfing scheme \cite{gross2011recovering} and least-square scheme, and our work makes no exception. However, due to the different setting in the sampling mechanism, the first part of the dual certificate (corresponding to Lemma \ref{Le: Certificate1}) requires an adapted construction technique. Though the construction of the second part (corresponding to Lemma \ref{Le: Certificate1}) remains the same, it is non-trivial for us to show that these properties are still satisfied in the deterministic case.
\begin{lemma}
	\textit{Subject to the settings of Thm \ref{Th: Main Theorem}, there exists $\boldsymbol{\Lambda}_L\in\mathcal{N}$ such that with very high probability,
		\begin{equation}
			\begin{aligned}
				&\left\lVert\mathcal{P}_{\mathcal{T}}\left[\boldsymbol{\Lambda}_L\right]-\boldsymbol{U}\boldsymbol{V}^\ast\right\lVert_F< \frac{1}{n^2},\\
				&\left\lVert\mathcal{P}_{\mathcal{T}^\bot}\left[\boldsymbol{\Lambda}_L\right]\right\lVert<\frac{1}{4},\\
				&\left\lVert\boldsymbol{\Lambda}_L\right\lVert_\infty<\frac{\lambda}{4}.
			\end{aligned}
			\nonumber
	\end{equation}}
	\label{Le: Certificate1}
\end{lemma}
\begin{proof}
	\textit{Deferred to Section \ref{sec:Mathematical Proofs}.}
\end{proof}
\begin{lemma}
	Conditioned on the settings of Thm \ref{Th: Main Theorem}, there exists $\boldsymbol{\Lambda}_S\in\mathcal{O}$ such that with very high probability,
	\begin{equation}
		\begin{aligned}
			&\mathcal{P}_{\mathcal{T}}\left[\boldsymbol{\Lambda}_S\right]=0,
			\,&\left\lVert\mathcal{P}_{\mathcal{T}^\bot}\left[\boldsymbol{\Lambda}_S\right]\right\lVert_F< \frac{1}{4},
			\\&\mathcal{P}_{\mathcal{V}}\left[\boldsymbol{\Lambda}_S\right]=\lambda\overline{\boldsymbol{\Sigma}}_0,\,&\left\lVert\mathcal{P}_{\mathcal{N}}\left[\boldsymbol{\Lambda}_S\right]\right\lVert_\infty<\frac{\lambda}{4}.
		\end{aligned}
		\nonumber
	\end{equation}
	\label{Le: Certificate2}
\end{lemma}
\begin{proof}
	\textit{Deferred to Section \ref{sec:Mathematical Proofs}.}
\end{proof}
\noindent
With Lemma \ref{Le: Certificate1} and \ref{Le: Certificate2}, if we denote $\boldsymbol{\Lambda}=\boldsymbol{\Lambda}_L+\boldsymbol{\Lambda}_S$, it follows that
\begin{equation}
	\begin{aligned}
		\left\lVert\mathcal{P}_{\mathcal{T}}\left[\boldsymbol{\Lambda}\right]-\boldsymbol{U}\boldsymbol{V}^\ast\right\lVert_F&=\left\lVert\mathcal{P}_{\mathcal{T}}\left[\boldsymbol{\Lambda}_L\right]-\boldsymbol{U}\boldsymbol{V}^\ast\right\lVert_F< \frac{1}{n^2},\\\left\lVert\mathcal{P}_{\mathcal{T}^\bot}\left[\boldsymbol{\Lambda}\right]\right\lVert_F&=\left\lVert\mathcal{P}_{\mathcal{T}^\bot}\left[\boldsymbol{\Lambda}_L+\boldsymbol{\Lambda}_S\right]\right\lVert_F\\&\leq\left\lVert\mathcal{P}_{\mathcal{T}^\bot}\left[\boldsymbol{\Lambda}_L\right]\right\lVert_F+\left\lVert\mathcal{P}_{\mathcal{T}^\bot}\left[\boldsymbol{\Lambda}_S\right]\right\lVert_F\\&<\frac{1}{2},\\\mathcal{P}_{\mathcal{V}}\left[\boldsymbol{\Lambda}\right]&=\mathcal{P}_{\mathcal{V}}\left[\boldsymbol{\Lambda}_S\right]=\lambda\overline{\boldsymbol{\Sigma}}_0,\\\left\lVert\mathcal{P}_{\mathcal{N}}\left[\boldsymbol{\Lambda}\right]\right\lVert_\infty&=\left\lVert\mathcal{P}_{\mathcal{N}}\left[\boldsymbol{\Lambda}_L+\boldsymbol{\Lambda}_S\right]\right\lVert_\infty\\&\leq\left\lVert\mathcal{P}_{\mathcal{N}}\left[\boldsymbol{\Lambda}_L\right]\right\lVert_\infty+\left\lVert\mathcal{P}_{\mathcal{N}}\left[\boldsymbol{\Lambda}_S\right]\right\lVert_\infty\\&<\frac{\lambda}{2}.
	\end{aligned}
	\nonumber
\end{equation}
Thus \eqref{eq:Relaxed KKT} is satisfied and Thm \ref{Th: Main Theorem} is proven.
\section{Mathematical Proofs}
\label{sec:Mathematical Proofs}
\subsection{Preparation Theories}
To fulfill the proof of those main theories in Section \ref{sec:Identifiability}, we need first to prepare some basic conclusions, which are delivered by this section.
\begin{proposition}
	\textit{Suppose $\mathsf{M}\thicksim Ber(\pi)$ is a Bernoulli random subsets and $\mathcal{M}$ the corresponding linear space accommodating those matrices supported on $\mathsf{M}$. Let $\mathcal{O}$ and $\boldsymbol{L}_0\in\mathbb{R}^{n\times n}$ be configured as Thm \ref{Th: Main Theorem}. Fix any constant $\epsilon>0$. Then there exists a numerical constant $C_a$ such that if $\pi>C_a\frac{\nu r\log(n)}{\epsilon^2n},$ with high probability,
		$$\left\lVert\mathcal{P}_{\mathcal{T}}-\pi^{-1}\mathcal{P}_\mathcal{T}\mathcal{P}_\mathcal{O}\mathcal{P}_\mathcal{M}\mathcal{P}_\mathcal{T}\right\lVert\leq \frac{1}{2}+\epsilon.$$
	}
	\label{Prop:Aux1}
\end{proposition}
\begin{proof}
	\textit{To begin, we consider to apply Bernstein inequality to bound the operator norm of 
		\begin{equation}
			\begin{aligned}
				\mathcal{P}_\mathcal{T}\mathcal{P}_\mathcal{O}\mathcal{P}_\mathcal{T}&-\pi^{-1}\mathcal{P}_\mathcal{T}\mathcal{P}_\mathcal{O}\mathcal{P}_\mathcal{M}\mathcal{P}_\mathcal{T}\\&
				\begin{matrix}
					=\sum_{(i,j)\in\mathsf{O}}&\underbrace{\mathcal{P}_\mathcal{T}\left(1-\pi^{-1}\delta_{ij}\right)\boldsymbol{E}_{ij}\left\langle\boldsymbol{E}_{ij},\cdot\right\rangle\mathcal{P}_\mathcal{T}}\\&\mathcal{H}_{ij}
				\end{matrix}
			\end{aligned}
			\nonumber
		\end{equation}
		where $\delta_{ij}$ are independent Bernoulli variables with the same distribution $Pr(\delta_{ij}=1)=\pi,\,Pr(\delta_{ij}=0)=1-\pi.$
		Since $\mathcal{E}\left[\sum_{(i,j)\in\mathsf{O}}\mathcal{H}_{ij}\right]=0$, we further seek to bound $\max_{ij}\left\lVert\mathcal{H}_{ij}\right\lVert$ and the 'variance' $\sum_{(i,j)\in\mathsf{O}}\mathcal{E}\left[\mathcal{H}_{ij}^\ast\mathcal{H}_{ij}\right].$
		\begin{itemize}
			\item
			A bound on the summands with probability 1:
			\begin{align}
				\left\lVert\mathcal{H}_{ij}\right\lVert&=\left(1-\pi^{-1}\delta_{ij}\right)\left\lVert\mathcal{P}_\mathcal{T}\boldsymbol{E}_{ij}\left\langle\boldsymbol{E}_{ij},\cdot\right\rangle\mathcal{P}_\mathcal{T}\right\lVert\notag\\&\leq \pi^{-1}\left\lVert\mathcal{P}_\mathcal{T}\boldsymbol{E}_{ij}\left\langle\boldsymbol{E}_{ij},\cdot\right\rangle\mathcal{P}_\mathcal{T}\right\lVert\notag\\&=\pi^{-1}\left\lVert\mathcal{P}_\mathcal{T}\left[\boldsymbol{E}_{ij}\right]\left\langle\mathcal{P}_\mathcal{T}\left[\boldsymbol{E}_{ij}\right],\cdot\right\rangle\right\lVert\notag\\&=\pi^{-1}\left\lVert\mathcal{P}_\mathcal{T}\left[\boldsymbol{E}_{ij}\right]\right\lVert_F^2\notag\\&\leq\frac{2\nu r}{n\pi}\notag\\&\leq \frac{2\epsilon^2}{C_a\log(n)}.
			\end{align}
			\item 
			A bound on the 'variance':
			\begin{align}
				&\quad\left\lVert\sum_{(i,j)\in\mathsf{O}}\mathcal{E}\left[\mathcal{H}_{ij}^\ast\mathcal{H}_{ij}\right]\right\lVert
				=\left\lVert\sum_{(i,j)\in\mathsf{O}}\mathcal{E}\left[\mathcal{H}_{ij}\mathcal{H}_{ij}^\ast\right]\right\lVert\notag\\&=\Bigg\lVert\sum_{(i,j)\in\mathsf{O}}\mathcal{E}\bigg[\left(1-\pi^{-1}\delta_{ij}\right)^2\mathcal{P}_\mathcal{T}\boldsymbol{E}_{ij}\left\lVert\mathcal{P}_\mathcal{T}\left[\boldsymbol{E}_{ij}\right]\right\lVert_F^2\notag\\&\left\langle\boldsymbol{E}_{ij},\cdot\right\rangle\mathcal{P}_\mathcal{T}\bigg]\Bigg\lVert\notag\\&=\Bigg\lVert\sum_{(i,j)\in\mathsf{O}}\mathcal{E}\bigg[\left(1-\pi^{-1}\delta_{ij}\right)^2\bigg]\mathcal{P}_\mathcal{T}\boldsymbol{E}_{ij}\left\lVert\mathcal{P}_\mathcal{T}\left[\boldsymbol{E}_{ij}\right]\right\lVert_F^2\notag\\&\left\langle\boldsymbol{E}_{ij},\cdot\right\rangle\mathcal{P}_\mathcal{T}\Bigg\lVert\notag\\&\leq\pi^{-1}\left\lVert\sum_{(i,j)\in\mathsf{O}}\mathcal{P}_\mathcal{T}\boldsymbol{E}_{ij}\left\lVert\mathcal{P}_\mathcal{T}\left[\boldsymbol{E}_{ij}\right]\right\lVert_F^2\left\langle\boldsymbol{E}_{ij},\cdot\right\rangle\mathcal{P}_\mathcal{T}\right\lVert\notag\\&\leq\frac{2\nu r}{n\pi}\left\lVert\sum_{(i,j)\in\mathsf{O}}\mathcal{P}_\mathcal{T}\boldsymbol{E}_{ij}\left\langle\boldsymbol{E}_{ij},\cdot\right\rangle\mathcal{P}_\mathcal{T}\right\lVert\notag\\&\leq\frac{2\epsilon^2}{C_a\log(n)}\left\lVert\mathcal{P}_\mathcal{T}\mathcal{P}_\mathcal{O}\mathcal{P}_\mathcal{T}\right\lVert\notag\\&\leq\frac{2\epsilon^2}{C_a\log(n)}.\notag
			\end{align}
		\end{itemize}
		Thereupon, we're able to invoke Bernstein inequality to obtain that
		\begin{equation}
			Pr\left(\left\lVert\sum_{(i,j)\in\mathsf{O}}\mathcal{H}_{ij}\right\lVert>\epsilon\right)\leq2n^{1-\frac{3C_a}{12+4\epsilon}}.
			\nonumber
		\end{equation}
		Finally, we complete the proof by indicating that with probability at least $1-2n^{1-\frac{3C_a}{12+4\epsilon}}$
		\begin{align}
			&\left\lVert\mathcal{P}_{\mathcal{T}}-\pi^{-1}\mathcal{P}_\mathcal{T}\mathcal{P}_\mathcal{O}\mathcal{P}_\mathcal{M}\mathcal{P}_\mathcal{T}\right\lVert
			\notag\\&=\left\lVert\sum_{(i,j)\in\mathsf{O}}\mathcal{H}_{ij}+\mathcal{P}_\mathcal{T}\mathcal{P}_\mathcal{O}^\bot\mathcal{P}_\mathcal{T}\right\lVert\notag\\&\leq\left\lVert\sum_{(i,j)\in\mathsf{O}}\mathcal{H}_{ij}\right\lVert+\left\lVert\mathcal{P}_\mathcal{T}\mathcal{P}_\mathcal{O}^\bot\mathcal{P}_\mathcal{T}\right\lVert\notag\\&\leq \epsilon +2\left(1-\gamma_{\mathcal{O},\mathcal{O}^T}\left(\boldsymbol{L}_0\right)\right)\notag\\&\leq \epsilon+\frac{1}{2}.\notag
		\end{align}
	}
\end{proof}
\begin{corollary}
	\textit{Conditioned on the settings of Thm \ref{Th: Main Theorem}, with very high probability, $\left\lVert\mathcal{P}_{\mathcal{V}}\mathcal{P}_{\mathcal{T}}\right\lVert\leq \sqrt{\rho+\epsilon'}$ where $\epsilon'$ is an arbitrarily small constant.}
	\label{Cor:Aux1}
\end{corollary}
\begin{proof}
	The proof of Prop \ref{Prop:Aux1} gives that
	$$\left\lVert\mathcal{P}_\mathcal{T}\mathcal{P}_\mathcal{O}\mathcal{P}_\mathcal{T}-(1-\rho)^{-1}\mathcal{P}_\mathcal{T}\mathcal{P}_\mathcal{O}\mathcal{P}_{\mathcal{W}^\bot}\mathcal{P}_\mathcal{T}\right\lVert\leq \epsilon.$$
	Thus,
	\begin{align}
		&\quad\left\lVert\mathcal{P}_{\mathcal{V}}\mathcal{P}_{\mathcal{T}}\right\lVert^2=\left\lVert\mathcal{P}_{\mathcal{T}}\mathcal{P}_{\mathcal{V}}\mathcal{P}_{\mathcal{T}}\right\lVert\notag\\&=\left\lVert\mathcal{P}_{\mathcal{T}}\mathcal{P}_{\mathcal{O}}\left(\mathcal{P}_{\mathcal{I}}-\mathcal{P}_{\mathcal{W}^\bot}\right)\mathcal{P}_{\mathcal{T}}\right\lVert\notag\\&\leq\left\lVert(1-\rho)\mathcal{P}_{\mathcal{T}}\mathcal{P}_{\mathcal{O}}\mathcal{P}_{\mathcal{T}}-\mathcal{P}_{\mathcal{T}}\mathcal{P}_{\mathcal{O}}\mathcal{P}_{\mathcal{W}^\bot}\mathcal{P}_{\mathcal{T}}\right\lVert+\rho\left\lVert\mathcal{P}_{\mathcal{T}}\mathcal{P}_{\mathcal{O}}\mathcal{P}_{\mathcal{T}}\right\lVert\notag\\&\leq(1-\rho)\epsilon+\rho=\rho+\epsilon'.
	\end{align}
\end{proof}
\begin{proposition}
	\textit{Suppose $\mathsf{M}\thicksim Ber(\pi)$ is a Bernoulli random subsets and $\mathcal{M}$ the corresponding linear space accommodating those matrices supported on $\mathsf{M}$. Let $\mathcal{O}$ and $\boldsymbol{L}_0\in\mathbb{R}^{n\times n}$ be configured as Thm \ref{Th: Main Theorem}. Fix any constant $\epsilon>0$. Then there exists positive numerical constants $C'_a$ and $\gamma<1$ such that if $\pi>C'_a\frac{\nu r\log(n)}{\epsilon^2n},$ with high probability,
		$$\left\lVert\left(\mathcal{I}-\pi^{-1}\mathcal{P}_\mathcal{T}\mathcal{P}_\mathcal{O}\mathcal{P}_\mathcal{M}\right)\left[\boldsymbol{D}\right]\right\lVert_\infty\leq \left(\gamma+\epsilon\right)\left\lVert\boldsymbol{D}\right\lVert_\infty$$
		for any $\boldsymbol{D}\in\mathcal{T}$.
	}
	\label{Prop:Aux2}
\end{proposition}
\begin{proof}
	\textit{Firstly, we apply Bernstein inequality to bound the infinity norm of 
		\begin{equation}
			\begin{aligned}
				\boldsymbol{D}'&\triangleq\left(\mathcal{P}_\mathcal{T}\mathcal{P}_\mathcal{O}-\pi^{-1}\mathcal{P}_\mathcal{T}\mathcal{P}_\mathcal{O}\mathcal{P}_\mathcal{M}\right)\left[\boldsymbol{D}\right]\\&=\sum_{(i,j)\in\mathsf{O}}{\left(1-\pi^{-1}\delta_{ij}\right){D}_{ij}\mathcal{P}_\mathcal{T}\left[\boldsymbol{E}_{ij}\right]}.
			\end{aligned}
			\nonumber
		\end{equation}
		where $\delta_{ij}$ are independent Bernoulli variables with the same distribution $Pr(\delta_{ij}=1)=\pi,\,Pr(\delta_{ij}=0)=1-\pi.$
		To do so, we consider each element of it, \textit{i.e.,}
		\begin{equation}
			\begin{aligned}
				\begin{matrix}
					{D}'_{i_0j_0}=\sum_{(i,j)\in\mathsf{O}}&\underbrace{{\left(1-\pi^{-1}\delta_{ij}\right){D}_{ij}\left\langle\mathcal{P}_\mathcal{T}\left[\boldsymbol{E}_{ij}\right],\boldsymbol{E}_{i_0j_0}\right\rangle}}\\&{H}_{ij}.
				\end{matrix}.
			\end{aligned}
			\nonumber
		\end{equation}
		Then we have $\mathcal{E}\left[{H}_{ij}\right]=0$ and
		\begin{itemize}
			\item
			a bound on $\left\lvert{H}_{ij}\right\lvert$ with probability 1:
			\begin{equation}
				\begin{aligned}
					\left\lvert {H}_{ij}\right\lvert&=\left\lvert {\left(1-\pi^{-1}\delta_{ij}\right){D}_{ij}\left\langle\mathcal{P}_\mathcal{T}\left[\boldsymbol{E}_{ij}\right],\boldsymbol{E}_{i_0j_0}\right\rangle}\right\lvert\\&\leq \pi^{-1}\left\lVert\boldsymbol{D}\right\lVert_\infty\left\langle\mathcal{P}_\mathcal{T}\left[\boldsymbol{E}_{ij}\right],\mathcal{P}_\mathcal{T}\left[\boldsymbol{E}_{i_0j_0}\right]\right\rangle\\&\leq\pi^{-1}\left\lVert\boldsymbol{D}\right\lVert_\infty\max\left\{\left\lVert\mathcal{P}_\mathcal{T}\left[\boldsymbol{E}_{ij}\right]\right\lVert_F^2,\left\lVert\mathcal{P}_\mathcal{T}\left[\boldsymbol{E}_{i_0j_0}\right]\right\lVert_F^2\right\}\\&\leq \frac{\epsilon^2n}{C\nu r\log(n)}{\frac{2\nu r}{n}}\left\lVert\boldsymbol{D}\right\lVert_\infty\\&=\frac{2\epsilon^2}{C\log(n)}\left\lVert\boldsymbol{D}\right\lVert_\infty.
				\end{aligned}
				\nonumber
			\end{equation}
			\item 
			a bound on the sum of variance:
			\begin{equation}
				\begin{aligned}
					&\quad\sum_{(i,j)\in\mathsf{O}}\mathcal{E}\left[{H}_{ij}^2\right]\\&=\sum_{(i,j)\in\mathsf{O}}\mathcal{E}\left[{\left(1-\pi^{-1}\delta_{ij}\right)^2{D}_{ij}^2\left\langle\mathcal{P}_\mathcal{T}\left[\boldsymbol{E}_{ij}\right],\boldsymbol{E}_{i_0j_0}\right\rangle^2}\right]\\&\leq\pi^{-1}\left\lVert\boldsymbol{D}\right\lVert_\infty^2\sum_{(i,j)\in\mathsf{O}}\left\langle\boldsymbol{E}_{ij},\mathcal{P}_\mathcal{T}\left[\boldsymbol{E}_{i_0j_0}\right]\right\rangle^2\\&\leq\pi^{-1}\left\lVert\boldsymbol{D}\right\lVert_\infty^2\left\lVert\mathcal{P}_\mathcal{T}\left[\boldsymbol{E}_{i_0j_0}\right]\right\lVert_F^2\\&\leq \frac{\epsilon^2n}{C'_a\nu r\log(n)}{\frac{2\nu r}{n}}\left\lVert\boldsymbol{D}\right\lVert_\infty^2\\&=\frac{2\epsilon^2}{C'_a\log(n)}\left\lVert\boldsymbol{D}\right\lVert_\infty^2.
				\end{aligned}
				\nonumber
			\end{equation}
		\end{itemize}
		Hence, we invoke Bernstein inequality to obtain that
		\begin{equation}
			\begin{aligned}
				&\quad Pr\left(\left\lVert\boldsymbol{D}'\right\lVert_\infty>\epsilon\left\lVert\boldsymbol{D}\right\lVert_\infty\right)\\&\leq \sum_{i_0j_0}Pr\left(\left\lvert\sum_{(i,j)\in\mathsf{O}} H_{ij}\right\lvert>\epsilon\left\lVert\boldsymbol{D}\right\lVert_\infty\right)\\&\leq 2n^2\exp\left(-\frac{\frac{\epsilon^2\left\lVert\boldsymbol{D}\right\lVert_\infty^2}{2}}{\frac{2\epsilon^2}{C'_a\log(n)}\left\lVert\boldsymbol{D}\right\lVert_\infty^2+\frac{\frac{2\epsilon^2}{C'_a\log(n)}\left\lVert\boldsymbol{D}\right\lVert_\infty\epsilon\left\lVert\boldsymbol{D}\right\lVert_\infty}{3}}\right)\\&=2n^{2-\frac{3C'_a}{12+4\epsilon}}.
			\end{aligned}
			\nonumber
		\end{equation}
		Recall that $\mathcal{P}_\mathcal{O}$ satisfies RAIIP, thus there exists $\gamma<1$, scu h that $\left\lVert\mathcal{P}_\mathcal{T}\mathcal{P}_\mathcal{O}^\bot\left[\boldsymbol{D}\right]\right\lVert_\infty\leq\gamma \left\lVert\boldsymbol{D}\right\lVert_\infty.$ 
		It then follows that with probability at least $1-2n^{2-\frac{3C'_a}{12+4\epsilon}}$,
		\begin{equation}
			\begin{aligned}
				&\quad\left\lVert\left(\mathcal{I}-\pi^{-1}\mathcal{P}_\mathcal{T}\mathcal{P}_\mathcal{O}\mathcal{P}_\mathcal{M}\right)\left[\boldsymbol{D}\right]\right\lVert_\infty\\&=\left\lVert\left(\mathcal{P}_\mathcal{T}\mathcal{P}_\mathcal{O}-\pi^{-1}\mathcal{P}_\mathcal{T}\mathcal{P}_\mathcal{O}\mathcal{P}_\mathcal{M}\right)\left[\boldsymbol{D}\right]+\mathcal{P}_\mathcal{T}\mathcal{P}_\mathcal{O}^\bot\left[\boldsymbol{D}\right]\right\lVert_\infty\\&\leq\left\lVert\left(\mathcal{P}_\mathcal{T}\mathcal{P}_\mathcal{O}-\pi^{-1}\mathcal{P}_\mathcal{T}\mathcal{P}_\mathcal{O}\mathcal{P}_\mathcal{M}\right)\left[\boldsymbol{D}\right]\right\lVert_\infty+\left\lVert\mathcal{P}_\mathcal{T}\mathcal{P}_\mathcal{O}^\bot\left[\boldsymbol{D}\right]\right\lVert_\infty\\&\leq\left(\gamma+\epsilon\right)\left\lVert\boldsymbol{D}\right\lVert_\infty.
			\end{aligned}
			\nonumber
		\end{equation}
	}
\end{proof}
\begin{proposition}
	\textit{Under the assumption of Thm {\ref{Th: Main Theorem}}, with probability high enough, we have that for any matrix $\boldsymbol{P}\in\mathbb{R}^{n\times n}$, $$\left\lVert\mathcal{P}_{\mathcal{T}}\left[\boldsymbol{P}\right]\right\lVert_F\leq(n+1)\left\lVert\mathcal{P}_{\mathcal{N}}\left[\boldsymbol{P}\right]\right\lVert_F+n\left\lVert\mathcal{P}_{\mathcal{T}^\bot}\left[\boldsymbol{P}\right]\right\lVert_F$$ and $\mathcal{N}^\bot\cap\mathcal{T}=\left\{\boldsymbol{0}\right\}.$}
	\label{Prop:Aux3}
\end{proposition}
\begin{proof}
	Firstly, we have that
	\begin{align}
		\left\lVert\mathcal{P}_{\mathcal{N}}\mathcal{P}_{\mathcal{T}}\mathcal{P}_{\mathcal{N}^\bot}\left[\boldsymbol{P}\right]\right\lVert_F&=\left\lVert\mathcal{P}_{\mathcal{N}}\left(\mathcal{I}-\mathcal{P}_{\mathcal{T}^\bot}\right)\mathcal{P}_{\mathcal{N}^\bot}\left[\boldsymbol{P}\right]\right\lVert_F\notag\\&=\left\lVert\mathcal{P}_{\mathcal{N}}\mathcal{P}_{\mathcal{T}^\bot}\mathcal{P}_{\mathcal{N}^\bot}\left[\boldsymbol{P}\right]\right\lVert_F\notag\\&\leq\left\lVert\mathcal{P}_{\mathcal{T}^\bot}\mathcal{P}_{\mathcal{N}^\bot}\left[\boldsymbol{P}\right]\right\lVert_F. 
	\end{align}
	Note that by Prop \ref{Prop:Aux1}, $\left\lVert\mathcal{P}_{\mathcal{T}}-(1-\rho)^{-1}\mathcal{P}_\mathcal{T}\mathcal{P}_\mathcal{N}\mathcal{P}_\mathcal{T}\right\lVert\leq \frac{1}{2}+\epsilon$. It is further obtained that
	\begin{align}
		&\quad\left\lVert\mathcal{P}_{\mathcal{N}}\mathcal{P}_{\mathcal{T}}\mathcal{P}_{\mathcal{N}^\bot}\left[\boldsymbol{P}\right]\right\lVert_F^2\notag\\&=\left\langle\mathcal{P}_{\mathcal{N}}\mathcal{P}_{\mathcal{T}}\mathcal{P}_{\mathcal{N}^\bot}\left[\boldsymbol{P}\right],\mathcal{P}_{\mathcal{N}}\mathcal{P}_{\mathcal{T}}\mathcal{P}_{\mathcal{N}^\bot}\left[\boldsymbol{P}\right]\right\rangle\notag\\&=\left\langle\mathcal{P}_{\mathcal{T}}\mathcal{P}_{\mathcal{N}^\bot}\left[\boldsymbol{P}\right],\mathcal{P}_{\mathcal{T}}\mathcal{P}_{\mathcal{N}}\mathcal{P}_{\mathcal{T}}\mathcal{P}_{\mathcal{N}^\bot}\left[\boldsymbol{P}\right]\right\rangle\notag\\&=(1-\rho)\left\langle\mathcal{P}_{\mathcal{T}}\mathcal{P}_{\mathcal{N}^\bot}\left[\boldsymbol{P}\right],(1-\rho)^{-1}\mathcal{P}_{\mathcal{T}}\mathcal{P}_{\mathcal{N}}\mathcal{P}_{\mathcal{T}}\mathcal{P}_{\mathcal{N}^\bot}\left[\boldsymbol{P}\right]\right\rangle\notag\\&=(1-\rho)\left\langle\mathcal{P}_{\mathcal{T}}\mathcal{P}_{\mathcal{N}^\bot}\left[\boldsymbol{P}\right],\mathcal{P}_{\mathcal{T}}\mathcal{P}_{\mathcal{N}^\bot}\left[\boldsymbol{P}\right]\right\rangle\notag\\&\quad+(1-\rho)\big\langle\mathcal{P}_{\mathcal{T}}\mathcal{P}_{\mathcal{N}^\bot}\left[\boldsymbol{P}\right],\notag\\&\quad\quad\left((1-\rho)^{-1}\mathcal{P}_\mathcal{T}\mathcal{P}_\mathcal{N}\mathcal{P}_\mathcal{T}-\mathcal{P}_{\mathcal{T}}\right)\mathcal{P}_{\mathcal{N}^\bot}\left[\boldsymbol{P}\right]\big\rangle\notag\\&\geq(1-\rho)\left(\left\lVert\mathcal{P}_{\mathcal{T}}\mathcal{P}_{\mathcal{N}^\bot}\left[\boldsymbol{P}\right]\right\lVert_F^2-\left(\frac{1}{2}+\epsilon\right)\left\lVert\mathcal{P}_{\mathcal{T}}\mathcal{P}_{\mathcal{N}^\bot}\left[\boldsymbol{P}\right]\right\lVert_F^2\right)\notag\\&=(1-\rho)\left(\frac{1}{2}-\epsilon\right)\left\lVert\mathcal{P}_{\mathcal{T}}\mathcal{P}_{\mathcal{N}^\bot}\left[\boldsymbol{P}\right]\right\lVert_F^2\notag\\&\geq\frac{1}{n^2}\left\lVert\mathcal{P}_{\mathcal{T}}\mathcal{P}_{\mathcal{N}^\bot}\left[\boldsymbol{P}\right]\right\lVert_F^2
	\end{align}
	provided small enough $\left(\rho,\epsilon\right)$ or large enough $n$. Hence $$\left\lVert\mathcal{P}_{\mathcal{T}}\mathcal{P}_{\mathcal{N}^\bot}\left[\boldsymbol{P}\right]\right\lVert_F\leq n\left\lVert\mathcal{P}_{\mathcal{T}^\bot}\mathcal{P}_{\mathcal{N}^\bot}\left[\boldsymbol{P}\right]\right\lVert_F,$$ and $\mathcal{N}^\bot\cap\mathcal{T}=\left\{\boldsymbol{0}\right\}$ follows. Eventually, we observe that
	\begin{align}
		&\quad\left\lVert\mathcal{P}_{\mathcal{T}}\left[\boldsymbol{P}\right]\right\lVert_F\notag\\&\leq\left\lVert\mathcal{P}_{\mathcal{T}}\mathcal{P}_{\mathcal{N}}\left[\boldsymbol{P}\right]\right\lVert_F+\left\lVert\mathcal{P}_{\mathcal{T}}\mathcal{P}_{\mathcal{N}^\bot}\left[\boldsymbol{P}\right]\right\lVert_F\notag\\&\leq\left\lVert\mathcal{P}_{\mathcal{T}}\mathcal{P}_{\mathcal{N}}\left[\boldsymbol{P}\right]\right\lVert_F+n\left\lVert\mathcal{P}_{\mathcal{T}^\bot}\mathcal{P}_{\mathcal{N}^\bot}\left[\boldsymbol{P}\right]\right\lVert_F\notag\\&\leq\left\lVert\mathcal{P}_{\mathcal{T}}\mathcal{P}_{\mathcal{N}}\left[\boldsymbol{P}\right]\right\lVert_F+n\left(\left\lVert\mathcal{P}_{\mathcal{T}^\bot}\mathcal{P}_{\mathcal{N}}\left[\boldsymbol{P}\right]\right\lVert_F+\left\lVert\mathcal{P}_{\mathcal{T}^\bot}\left[\boldsymbol{P}\right]\right\lVert_F\right)\notag\\&\leq(n+1)\left\lVert\mathcal{P}_{\mathcal{N}}\left[\boldsymbol{P}\right]\right\lVert_F+n\left\lVert\mathcal{P}_{\mathcal{T}^\bot}\left[\boldsymbol{P}\right]\right\lVert_F.
	\end{align}
\end{proof}
\subsection{Proofs of Lemma \ref{Le: Relaxed KKT}-\ref{Le: Certificate2}}
\subsubsection{Proof of Lemma \ref{Le: Relaxed KKT}}
\begin{proof}
	Consider an arbitrary non-zero perturbation $\boldsymbol{P}\in\mathbb{R}^{n\times n}$, we construct $\boldsymbol{\Gamma}_L=\boldsymbol{U}\boldsymbol{V}^\ast+\overline{\boldsymbol{U}}\overline{\boldsymbol{V}}^\ast$ where $\overline{\boldsymbol{U}},\overline{\boldsymbol{V}}$ are obtained from the compact SVD of $\mathcal{P}_{\mathcal{T}^\bot}\left[\boldsymbol{P}\right],$ then $\left\langle\overline{\boldsymbol{U}}\overline{\boldsymbol{V}}^\ast,\mathcal{P}_{\mathcal{T}^\bot}\left[\boldsymbol{P}\right]\right\rangle=\left\lVert\mathcal{P}_{\mathcal{T}^\bot}\left[\boldsymbol{P}\right]\right\lVert_\ast.$ Let $\boldsymbol{\Gamma}_S=\lambda\left(\overline{\boldsymbol{\Sigma}}_0-sgn\left(\mathcal{P}_{\mathcal{N}}\left[\boldsymbol{P}\right]\right)\right)$. Note that by Prop \ref{Prop:Aux3}, $\mathcal{N}^\bot\cap\mathcal{T}=\left\{\boldsymbol{0}\right\}$ and $\left\lVert\mathcal{P}_{\mathcal{T}}\left[\boldsymbol{P}\right]\right\lVert_F\leq(n+1)\left\lVert\mathcal{P}_{\mathcal{N}}\left[\boldsymbol{P}\right]\right\lVert_F+n\left\lVert\mathcal{P}_{\mathcal{T}^\bot}\left[\boldsymbol{P}\right]\right\lVert_F.$ It follows that
	\begin{align}
		&\left\lVert\boldsymbol{L}_0+\boldsymbol{P}\right\lVert_\ast+\lambda\left\lVert\overline{\boldsymbol{S}}_0-\mathcal{P}_{\mathcal{O}}\left[\boldsymbol{P}\right]\right\lVert_1-\left\lVert\boldsymbol{L}_0\right\lVert_\ast-\lambda\left\lVert\overline{\boldsymbol{S}}_0\right\lVert_1\notag\\&\quad\geq\left\langle\boldsymbol{\Gamma}_L,\boldsymbol{P}\right\rangle+\left\langle\boldsymbol{\Gamma}_S,-\mathcal{P}_{\mathcal{O}}\left[\boldsymbol{P}\right]\right\rangle\notag\\&\quad=\left\langle\boldsymbol{\Gamma}_L,\boldsymbol{P}\right\rangle+\left\langle\boldsymbol{\Gamma}_S,-\boldsymbol{P}\right\rangle\notag\\&\quad=\left\langle\boldsymbol{\Gamma}_L-\boldsymbol{\Lambda},\boldsymbol{P}\right\rangle+\left\langle\boldsymbol{\Gamma}_S-\boldsymbol{\Lambda},-\boldsymbol{P}\right\rangle\notag\\&\quad=\left\langle\boldsymbol{U}\boldsymbol{V}^\ast-\mathcal{P}_{\mathcal{T}}\left[\boldsymbol{\Lambda}\right],\mathcal{P}_{\mathcal{T}}\left[\boldsymbol{P}\right]\right\rangle\notag\\&\quad\quad+\left\langle\overline{\boldsymbol{U}}\overline{\boldsymbol{V}}^\ast-\mathcal{P}_{\mathcal{T}^\bot}\left[\boldsymbol{\Lambda}\right],\mathcal{P}_{\mathcal{T}^\bot}\left[\boldsymbol{P}\right]\right\rangle-\left\langle\lambda\overline{\boldsymbol{\Sigma}}_0-\mathcal{P}_{\mathcal{V}}\left[\boldsymbol{\Lambda}\right],\boldsymbol{P}\right\rangle\notag\\&\quad\quad+\left\langle-\lambda sgn\left(\mathcal{P}_{\mathcal{N}}\left[\boldsymbol{P}\right]\right)-\mathcal{P}_{\mathcal{V}^\bot}\left[\boldsymbol{\Lambda}\right],-\boldsymbol{P}\right\rangle\notag\\&\quad\geq -\frac{1}{n^2}\left\lVert\mathcal{P}_{\mathcal{T}}\left[\boldsymbol{P}\right]\right\lVert_F+\left\lVert\mathcal{P}_{\mathcal{T}^\bot}\left[\boldsymbol{P}\right]\right\lVert_\ast-\left\lVert\mathcal{P}_{\mathcal{T}^\bot}\left[\boldsymbol{\Lambda}\right]\right\lVert\left\lVert\mathcal{P}_{\mathcal{T}^\bot}\left[\boldsymbol{P}\right]\right\lVert_\ast\notag\\&\quad\quad+\lambda\left\lVert\mathcal{P}_{\mathcal{N}}\left[\boldsymbol{P}\right]\right\lVert_1-\left\lVert\mathcal{P}_{\mathcal{N}}\left[\boldsymbol{\Lambda}\right]\right\lVert_\infty\left\lVert\mathcal{P}_{\mathcal{N}}\left[\boldsymbol{P}\right]\right\lVert_1\notag\\&\quad\geq -\frac{1}{n^2}\left\lVert\mathcal{P}_{\mathcal{T}}\left[\boldsymbol{P}\right]\right\lVert_F+(1-\frac{1}{2})\left\lVert\mathcal{P}_{\mathcal{T}^\bot}\left[\boldsymbol{P}\right]\right\lVert_\ast\notag\\&\quad\quad+\lambda(1-\frac{1}{2})\left\lVert\mathcal{P}_{\mathcal{N}}\left[\boldsymbol{P}\right]\right\lVert_1\notag\\&\quad\geq -\frac{n+1}{n^2}\left\lVert\mathcal{P}_{\mathcal{N}}\left[\boldsymbol{P}\right]\right\lVert_F-\frac{n}{n^2}\left\lVert\mathcal{P}_{\mathcal{T}^\bot}\left[\boldsymbol{P}\right]\right\lVert_F\notag\\&\quad\quad+(1-\frac{1}{2})\left\lVert\mathcal{P}_{\mathcal{T}^\bot}\left[\boldsymbol{P}\right]\right\lVert_F+\lambda(1-\frac{1}{2})\left\lVert\mathcal{P}_{\mathcal{N}}\left[\boldsymbol{P}\right]\right\lVert_F\notag\\&\quad=(\lambda(1-\frac{1}{2})-\frac{n+1}{n^2})\left\lVert\mathcal{P}_{\mathcal{N}}\left[\boldsymbol{P}\right]\right\lVert_F\notag\\&\quad\quad+(1-\frac{1}{2}-\frac{1}{n})\left\lVert\mathcal{P}_{\mathcal{T}^\bot}\left[\boldsymbol{P}\right]\right\lVert_F.
	\end{align}
	Then using $\mathcal{N}^\bot\cap\mathcal{T}=\left\{\boldsymbol{0}\right\}$, the proof is completed.
\end{proof}
\subsubsection{Proof of Lemma \ref{Le: Certificate1}}
\begin{proof}
	We construct the first part of the desired dual certificate in spirit of the golfing scheme \cite{gross2011recovering,candes2011robust,li2013compressed}. Firstly, we construct a series of $k$ random subsets $\mathsf{M}_i\thicksim Ber(\eta)\in\mathbb{R}^{n\times n},\,i\in[k],$ where $\eta=1-\rho^{\frac{1}{k}}$ such that $\mathsf{M}=\mathsf{W}^\bot=\cup_{i=1}^k\mathsf{M}_i$, and denote $\mathcal{M}_i$ to be the linear space consisting of those matrices supported by $\mathsf{M}_i$.  Then simple deduction gives $\eta\geq\frac{1-\rho}{k}.$ Thus if we set $k=C_g\log(n)$, there is $\eta\geq\frac{C}{C_g}\frac{\nu r\log(n)}{n}.$ Then by Prop \ref{Prop:Aux1}, with high probability
	$$\left\lVert\mathcal{P}_{\mathcal{T}}-\eta^{-1}\mathcal{P}_\mathcal{T}\mathcal{P}_\mathcal{O}\mathcal{P}_{\mathcal{M}_i}\mathcal{P}_\mathcal{T}\right\lVert\leq \frac{1}{2}+\epsilon,\,i=1,2,\cdots,k.$$
	We next construct a sequence of matrices $\boldsymbol{\Lambda}_0=0,\boldsymbol{\Lambda}_1,\cdots,\boldsymbol{\Lambda}_k$ through
	\begin{equation}
		\begin{aligned}
			\boldsymbol{D}_i=\mathcal{P}_{\mathcal{T}}\mathcal{P}_{\mathcal{O}}\mathcal{P}_{\mathcal{M}}\left[\boldsymbol{\Lambda}_i\right]-\boldsymbol{U}\boldsymbol{V}^\ast,i=0,1,\cdots,k,\\\boldsymbol{\Lambda}_i=\boldsymbol{\Lambda}_{i-1}-\eta^{-1}\mathcal{P}_{\mathcal{O}}\mathcal{P}_{\mathcal{M}_i}\left[\boldsymbol{D}_{i-1}\right],i=1,\cdots,k.
		\end{aligned}		
	\end{equation}
	then we have $\boldsymbol{D}_i\in\mathcal{T},\boldsymbol{\Lambda}_i\in\mathcal{O}\cap\mathcal{M},$ hence
	\begin{equation}
		\begin{aligned}
			\boldsymbol{D}_i&=\mathcal{P}_{\mathcal{T}}\mathcal{P}_{\mathcal{O}}\mathcal{P}_{\mathcal{M}}\left[\boldsymbol{\Lambda}_i\right]-\boldsymbol{U}\boldsymbol{V}^\ast\\&=\mathcal{P}_{\mathcal{T}}\mathcal{P}_{\mathcal{O}}\mathcal{P}_{\mathcal{M}}\left[\boldsymbol{\Lambda}_{i-1}\right]-\boldsymbol{U}\boldsymbol{V}^\ast-\eta\mathcal{P}_{\mathcal{T}}\mathcal{P}_{\mathcal{O}}\mathcal{P}_{\mathcal{M}_i}\left[\boldsymbol{D}_{i-1}\right]\\&=\boldsymbol{D}_{i-1}-\eta\mathcal{P}_{\mathcal{T}}\mathcal{P}_{\mathcal{O}}\mathcal{P}_{\mathcal{M}_i}\left[\boldsymbol{D}_{i-1}\right]\\&=\left(\mathcal{P}_{\mathcal{T}}-\eta^{-1}\mathcal{P}_\mathcal{T}\mathcal{P}_\mathcal{O}\mathcal{P}_{\mathcal{M}_i}\mathcal{P}_\mathcal{T}\right)\left[\boldsymbol{D}_{i-1}\right].
		\end{aligned}		
	\end{equation}
	It follows that
	\begin{equation}
		\begin{aligned}
			\left\lVert\boldsymbol{D}_k\right\lVert_F&=\left\lVert\left(\mathcal{P}_{\mathcal{T}}-\eta^{-1}\mathcal{P}_\mathcal{T}\mathcal{P}_\mathcal{O}\mathcal{P}_{\mathcal{M}_i}\mathcal{P}_\mathcal{T}\right)^k\left[\boldsymbol{U}\boldsymbol{V}^\ast\right]\right\lVert_F\\&\leq\left(\frac{1}{2}+\epsilon\right)^k\sqrt{r}\\&=\frac{\sqrt{r}}{n^{C_g\log(\frac{2}{1+2\epsilon})}}.
		\end{aligned}		
	\end{equation}
	Moreover, Prop \ref{Prop:Aux2} helps to conclude that $$\left\lVert\left(\mathcal{P}_\mathcal{T}-\eta^{-1}\mathcal{P}_\mathcal{T}\mathcal{P}_\mathcal{O}\mathcal{P}_\mathcal{M}\mathcal{P}_\mathcal{T}\right)\left[\boldsymbol{D}_i\right]\right\lVert_\infty\leq \left(\gamma+\epsilon\right)\left\lVert\boldsymbol{D}_i\right\lVert_\infty,\,\forall i,$$ resulting in that
	\begin{equation}
		\begin{aligned}
			\left\lVert\boldsymbol{\Lambda}_k\right\lVert_\infty&=\left\lVert\sum_{i=1}^{k}-\eta^{-1}\mathcal{P}_{\mathcal{O}}\mathcal{P}_{\mathcal{M}_i}\left[\boldsymbol{D}_{i-1}\right]\right\lVert_\infty\\&\leq \eta^{-1}\sum_{i=1}^{k}\left\lVert\boldsymbol{D}_{i-1}\right\lVert_\infty\\&\leq \eta^{-1}\sum_{i=1}^{k}\left\lVert\boldsymbol{U}\boldsymbol{V}^\ast\right\lVert_\infty(\gamma+\epsilon)^{i-1}\\&=\frac{\eta^{-1}}{1-(\gamma+\epsilon)}\left\lVert\boldsymbol{U}\boldsymbol{V}^\ast\right\lVert_\infty\\&\leq \frac{C_g\log(n)}{(1-(\gamma+\epsilon))(1-\rho)}\left\lVert\boldsymbol{U}\boldsymbol{V}^\ast\right\lVert_\infty\\&\leq \frac{C_g\log(n)}{(1-(\gamma+\epsilon))\sqrt{(1-\rho)}}\sqrt{\frac{n}{C\nu r\log^2(n)}}\left\lVert\boldsymbol{U}\boldsymbol{V}^\ast\right\lVert_\infty\\&= \frac{C_g\sqrt{n}}{(1-(\gamma+\epsilon))\sqrt{(1-\rho)C}\sqrt{\nu r}}\left\lVert\boldsymbol{U}\boldsymbol{V}^\ast\right\lVert_\infty\\&\leq \frac{C_g\sqrt{n}}{(1-(\gamma+\epsilon))\sqrt{(1-\rho)C}\sqrt{\nu r}}\frac{\sqrt{\nu r}}{n\sqrt{\log(n)}}\\&= \frac{C_g}{(1-(\gamma+\epsilon))\sqrt{(1-\rho)C}}\lambda.
		\end{aligned}		
	\end{equation}
	Carefully choosing $C_g$ and $C$, there follows $\left\lVert\boldsymbol{D}_k\right\lVert_F<\frac{1}{n^2}$ and $\left\lVert\boldsymbol{\Lambda}_k\right\lVert_\infty<\frac{\lambda}{4}.$
	
	In the sequel, we attempt to bound $\left\lVert\mathcal{P}_{\mathcal{T}^\bot}\left[\boldsymbol{\Lambda}_k\right]\right\lVert.$ From the construction of $\boldsymbol{\Lambda}_k$, we have
	\begin{equation}
		\begin{aligned}
			\boldsymbol{\Lambda}_k=\sum_{i=1}^{k}-\eta^{-1}\mathcal{P}_{\mathcal{O}}\mathcal{P}_{\mathcal{M}_i}\left[\boldsymbol{D}_{i-1}\right],
		\end{aligned}		
	\end{equation}
	thus
	\begin{equation}
		\begin{aligned}
			\mathcal{P}_{\mathcal{T}^\bot}\left[\boldsymbol{\Lambda}_k\right]&=\sum_{i=1}^{k}-\eta^{-1}\mathcal{P}_{\mathcal{T}^\bot}\mathcal{P}_{\mathcal{O}}\mathcal{P}_{\mathcal{M}_i}\left[\boldsymbol{D}_{i-1}\right]\\&=\sum_{i=1}^{k}\mathcal{P}_{\mathcal{T}^\bot}\left(\mathcal{P}_{\mathcal{T}}-\eta^{-1}\mathcal{P}_{\mathcal{O}}\mathcal{P}_{\mathcal{M}_i}\right)\left[\boldsymbol{D}_{i-1}\right].
		\end{aligned}		
	\end{equation}
	Then with high probability
	\begin{align}
		\left\lVert\mathcal{P}_{\mathcal{T}^\bot}\left[\boldsymbol{\Lambda}_k\right]\right\lVert&\leq\sum_{i=1}^{k}\left\lVert\mathcal{P}_{\mathcal{T}^\bot}\left(\mathcal{P}_{\mathcal{T}}-\eta^{-1}\mathcal{P}_{\mathcal{O}}\mathcal{P}_{\mathcal{M}_i}\right)\left[\boldsymbol{D}_{i-1}\right]\right\lVert\notag\\&\leq\sum_{i=1}^{k}\left\lVert\left(\mathcal{P}_{\mathcal{T}}-\eta^{-1}\mathcal{P}_{\mathcal{O}}\mathcal{P}_{\mathcal{M}_i}\right)\left[\boldsymbol{D}_{i-1}\right]\right\lVert\notag\\&=\sum_{i=1}^{k}\left\lVert\left(\mathcal{I}-\eta^{-1}\mathcal{P}_{\mathcal{O}}\mathcal{P}_{\mathcal{M}_i}\right)\left[\boldsymbol{D}_{i-1}\right]\right\lVert\notag\\&=\sum_{i=1}^{k}\left\lVert\left(\mathcal{P}_{\mathcal{O}}-\eta^{-1}\mathcal{P}_{\mathcal{O}}\mathcal{P}_{\mathcal{M}_i}+\mathcal{P}_{\mathcal{O}^\bot}\right)\left[\boldsymbol{D}_{i-1}\right]\right\lVert\notag\\&\leq\sum_{i=1}^{k}\Big(\left\lVert\left(\mathcal{P}_{\mathcal{O}}-\eta^{-1}\mathcal{P}_{\mathcal{O}}\mathcal{P}_{\mathcal{M}_i}\right)\left[\boldsymbol{D}_{i-1}\right]\right\lVert\notag\\&\quad+\left\lVert\mathcal{P}_{\mathcal{O}^\bot}\left[\boldsymbol{D}_{i-1}\right]\right\lVert\Big)\notag\\&\leq\sum_{i=1}^{k}\Big(\left\lVert\left(\mathcal{I}-\eta^{-1}\mathcal{P}_{\mathcal{M}_i}\right)\left[\boldsymbol{D}_{i-1}\right]\right\lVert+\left\lVert\boldsymbol{D}_{i-1}\right\lVert\Big)\notag\\&\leq \sum_{i=1}^{k}\Big(\left\lVert\left(\mathcal{I}-\eta^{-1}\mathcal{P}_{\mathcal{M}_i}\right)\left[\boldsymbol{D}_{i-1}\right]\right\lVert+n\left\lVert\boldsymbol{D}_{i-1}\right\lVert\Big)\notag\\&\leq \sum_{i=1}^{k}\Big(C'_g n\left\lVert\boldsymbol{D}_{i-1}\right\lVert_\infty+n\left\lVert\boldsymbol{D}_{i-1}\right\lVert_\infty\Big)\notag\\&=\sum_{i=1}^{k}(C'_g+1) n \left\lVert\boldsymbol{D}_{i-1}\right\lVert_\infty\notag\\&\leq \sum_{i=1}^{k}(C'_g+1)(C'_g+1)n\left(\gamma+\epsilon\right)^{i-1} \left\lVert\boldsymbol{U}\boldsymbol{V}^\ast\right\lVert_\infty\notag\\&\leq \frac{(C'_g+1) \sqrt{\nu r}}{(1-(\gamma+\epsilon))\sqrt{\log(n)}}
	\end{align}
	where the fourth to last row is obtained by invoking Thm 6.3 in \cite{WOS:000272299900003} and $C'_g$ is some positive constant. Hence with a large enough $n$, we have $\left\lVert\mathcal{P}_{\mathcal{T}^\bot}\left[\boldsymbol{\Lambda}_L\right]\right\lVert<\frac{1}{4}$. 
	Finally, by setting $\boldsymbol{\Lambda}_L=\mathcal{P}_{\mathcal{O}}\mathcal{P}_{\mathcal{M}}\left[\boldsymbol{\Lambda}_k\right]=\mathcal{P}_{\mathcal{N}}\left[\boldsymbol{\Lambda}_k\right]$, the first part of the  desired dual certificate is established.
\end{proof}
\subsubsection{Proof of Lemma \ref{Le: Certificate2}}
\begin{proof}
	Recall that Cor \ref{Cor:Aux1} gives $\left\lVert\mathcal{P}_{\mathcal{V}}\mathcal{P}_{\mathcal{T}}\right\lVert^2\leq \rho_s$ (omitting a small enough $\epsilon$ for simplicity). Besides, by Lemma 5.6 and 5.11 in \cite{liu2019matrix}, given $\gamma_{\mathsf{O},\mathsf{O}^T}\left(\boldsymbol{L}_0\right)>\frac{3}{4}$,  $\left\lVert\left(\mathcal{P}_{\mathcal{T}}\mathcal{P}_{\mathcal{O}}\mathcal{P}_{\mathcal{T}}\right)^{-1}\right\lVert=\left\lVert\sum_{i=0}^{\infty}\left(\mathcal{P}_{\mathcal{T}}\mathcal{P}_{\mathcal{O}^\bot}\mathcal{P}_{\mathcal{T}}\right)^i\right\lVert\leq\sum_{i=0}^{\infty}2^{-i}=2.$ This brings that $\left\lVert\mathcal{P}_{\mathcal{V}}\mathcal{P}_{(\mathcal{T}+\mathcal{O}^\bot)}\mathcal{P}_{\mathcal{V}}\right\lVert=\left\lVert\mathcal{P}_{\mathcal{V}}\mathcal{P}_{\mathcal{T}}\left(\mathcal{P}_{\mathcal{T}}\mathcal{P}_{\mathcal{O}}\mathcal{P}_{\mathcal{T}}\right)^{-1}\mathcal{P}_{\mathcal{T}}\mathcal{P}_{\mathcal{V}}\right\lVert\leq2\rho_s.$ Thus $\sum_{i=0}^{\infty}\left(\mathcal{P}_{\mathcal{V}}\mathcal{P}_{(\mathcal{T}+\mathcal{O}^\bot)}\mathcal{P}_{\mathcal{V}}\right)^i$ is well defined and we're able to consider $$\boldsymbol{\Lambda}_S\triangleq\lambda\left(\mathcal{I}-\mathcal{P}_{(\mathcal{T}+\mathcal{O}^\bot)}\right)\sum_{i=0}^{\infty}\left(\mathcal{P}_{\mathcal{V}}\mathcal{P}_{(\mathcal{T}+\mathcal{O}^\bot)}\mathcal{P}_{\mathcal{V}}\right)^i\left[\overline{\boldsymbol{\Sigma}}_0\right].$$ Obviously, $\mathcal{P}_{\mathcal{T}}\left[\boldsymbol{\Lambda}_S\right]=\mathcal{P}_{\mathcal{O}^\bot}\left[\boldsymbol{\Lambda}_S\right]=0$ and 
	\begin{equation}
		\begin{aligned}
			&\quad\mathcal{P}_{\mathcal{V}}\left[\boldsymbol{\Lambda}_S\right]\\&=\lambda\mathcal{P}_{\mathcal{V}}\left(\mathcal{I}-\mathcal{P}_{(\mathcal{T}+\mathcal{O}^\bot)}\right)\sum_{i=0}^{\infty}\left(\mathcal{P}_{\mathcal{V}}\mathcal{P}_{(\mathcal{T}+\mathcal{O}^\bot)}\mathcal{P}_{\mathcal{V}}\right)^i\left[\overline{\boldsymbol{\Sigma}}_0\right]\\&=\lambda\left(\sum_{i=0}^{\infty}\left(\mathcal{P}_{\mathcal{V}}\mathcal{P}_{(\mathcal{T}+\mathcal{O}^\bot)}\mathcal{P}_{\mathcal{V}}\right)^i-\sum_{i=1}^{\infty}\left(\mathcal{P}_{\mathcal{V}}\mathcal{P}_{(\mathcal{T}+\mathcal{O}^\bot)}\mathcal{P}_{\mathcal{V}}\right)^i\right)\left[\overline{\boldsymbol{\Sigma}}_0\right]\\&=\lambda\overline{\boldsymbol{\Sigma}}_0.
		\end{aligned}
	\end{equation}
	It then remains to bound $\left\lVert\mathcal{P}_{\mathcal{T}^\bot}\left[\boldsymbol{\Lambda}_S\right]\right\lVert$ and $\left\lVert\mathcal{P}_{\mathcal{N}}\left[\boldsymbol{\Lambda}_S\right]\right\lVert_\infty$.
	
	Firstly, we expand $\mathcal{P}_{\mathcal{T}^\bot}\left[\boldsymbol{\Lambda}_S\right]$ as
	\begin{align}
		\mathcal{P}_{\mathcal{T}^\bot}\left[\boldsymbol{\Lambda}_S\right]&=\lambda\left(\mathcal{I}-\mathcal{P}_{(\mathcal{T}+\mathcal{O}^\bot)}\right)\left[\overline{\boldsymbol{\Sigma}}_0\right]\notag\\&\quad+\lambda\left(\mathcal{I}-\mathcal{P}_{(\mathcal{T}+\mathcal{O}^\bot)}\right)\sum_{i=1}^{\infty}\left(\mathcal{P}_{\mathcal{V}}\mathcal{P}_{(\mathcal{T}+\mathcal{O}^\bot)}\mathcal{P}_{\mathcal{V}}\right)^i\left[\overline{\boldsymbol{\Sigma}}_0\right].
	\end{align}
	The operator norm of the first term can be controlled as
	\begin{align}
		&\quad\left\lVert\lambda\left(\mathcal{I}-\mathcal{P}_{(\mathcal{T}+\mathcal{O}^\bot)}\right)\left[\overline{\boldsymbol{\Sigma}}_0\right]\right\lVert\leq\lambda\left\lVert\overline{\boldsymbol{\Sigma}}_0\right\lVert\notag\\&=\lambda\left\lVert\mathcal{P}_{\mathcal{O}}\left[\boldsymbol{\Sigma}_0\right]\right\lVert\leq\lambda\left\lVert\boldsymbol{\Sigma}_0\right\lVert\leq4\lambda\sqrt{n\rho}\leq4\sqrt{\frac{\rho}{\log(n)}}
	\end{align}
	where $\boldsymbol{\Sigma}_0\triangleq sgn\left(\boldsymbol{S}_0\right)$ and the second to last inequality is given in \cite{vershynin2010introduction}. In regard to the second term, let $\mathsf{N}$ be a $\frac{1}{2}$-Net for $\mathbb{S}^{n-1}$ (see \cite{ledoux2001concentration} for the definition and existence of $\epsilon$-Net), basic arguments in \cite{vershynin2010introduction} delivers that
	\begin{align}
		\left\lVert\mathcal{F}\left[\overline{\boldsymbol{\Sigma}}_0\right]\right\lVert=\sup_{\boldsymbol{x},\boldsymbol{y}\in\mathbb{S}^{n-1}}\,\boldsymbol{x}^*\mathcal{F}\left[\overline{\boldsymbol{\Sigma}}_0\right]\boldsymbol{y}\leq4\sup_{\boldsymbol{x},\boldsymbol{y}\in\mathbb{N}}\,\boldsymbol{x}^*\mathcal{F}\left[\overline{\boldsymbol{\Sigma}}_0\right]\boldsymbol{y}\notag\\=4\sup_{\boldsymbol{x},\boldsymbol{y}\in\mathbb{N}}\,\left\langle\mathcal{F}\left[\boldsymbol{x}\boldsymbol{y}^*\right],\overline{\boldsymbol{\Sigma}}_0\right\rangle
	\end{align}
	where $\mathcal{F}\left[\cdot\right]\triangleq\lambda\left(\mathcal{I}-\mathcal{P}_{(\mathcal{T}+\mathcal{O}^\bot)}\right)\sum_{i=1}^{\infty}\left(\mathcal{P}_{\mathcal{V}}\mathcal{P}_{(\mathcal{T}+\mathcal{O}^\bot)}\mathcal{P}_{\mathcal{V}}\right)^i\left[\cdot\right].$
	Since $\left\langle\mathcal{F}\left[\boldsymbol{x}\boldsymbol{y}^*\right],\overline{\boldsymbol{\Sigma}}_0\right\rangle$ is a linear combination of Rademacher $\pm1$ random variables, we may 
	invoke Hoeffding's inequality to obtain that conditioned on some specific support $\mathsf{V}$,
	$$Pr\left(\left\langle\mathcal{F}\left[\boldsymbol{x}\boldsymbol{y}^*\right],\overline{\boldsymbol{\Sigma}}_0\right\rangle>t\mid \mathsf{V}\right)\leq\exp\left(-\frac{t^2}{2\left\lVert\mathcal{F}\left[\boldsymbol{x}\boldsymbol{y}^*\right]\right\lVert_F^2}\right).$$
	Since $\left\lVert\mathcal{F}\left[\boldsymbol{x}\boldsymbol{y}^*\right]\right\lVert_F\leq \lambda \frac{2\rho_s}{1-2\rho_s}$ conditioned on $\mathsf{E}=\left\{\left\lVert\mathcal{P}_{\mathcal{V}}\mathcal{P}_{\mathcal{T}}\right\lVert^2\leq \rho_s\right\}$, we conclude that
	\begin{align}
		&\quad Pr\left(\left\lVert\mathcal{F}\left[\overline{\boldsymbol{\Sigma}}_0\right]\right\lVert>t\right)\notag\\&\leq Pr\left(\sup_{\boldsymbol{x},\boldsymbol{y}\in\mathbb{N}}\,\left\langle\mathcal{F}\left[\boldsymbol{x}\boldsymbol{y}^*\right],\overline{\boldsymbol{\Sigma}}_0\right\rangle>\frac{t}{4}\right)\notag\\&\leq Pr\left(\sup_{\boldsymbol{x},\boldsymbol{y}\in\mathbb{N}}\,\left\langle\mathcal{F}\left[\boldsymbol{x}\boldsymbol{y}^*\right],\overline{\boldsymbol{\Sigma}}_0\right\rangle>\frac{t}{4}\mid\mathsf{E}\right)+Pr\left(\mathsf{E}^c\right)\notag\\&\leq\mid\mathsf{N}\mid^2\sup_{\boldsymbol{x},\boldsymbol{y}\in\mathbb{N}}Pr\left(\left\langle\mathcal{F}\left[\boldsymbol{x}\boldsymbol{y}^*\right],\overline{\boldsymbol{\Sigma}}_0\right\rangle>\frac{t}{4}\mid\mathsf{E}\right)+Pr\left(\mathsf{E}^c\right)\notag\\&\leq 6^{2n}\exp\left(-\frac{t^2(1-2\rho_s)^2}{128\rho^2_s\lambda^2}\right)+Pr\left(\mathsf{E}^c\right)\notag\\&=\left(\frac{36}{n^{\frac{t^2(1-2\rho_s)^2}{128\rho^2_s}}}\right)^n+Pr\left(\mathsf{E}^c\right).
	\end{align}
	Setting $t=\frac{1}{8}$, with small enough $\rho_s$ and large enough $n$ we will have that \begin{align}
		\left\lVert\lambda\left(\mathcal{I}-\mathcal{P}_{(\mathcal{T}+\mathcal{O}^\bot)}\right)\left[\overline{\boldsymbol{\Sigma}}_0\right]\right\lVert\leq\frac{1}{8},\notag\\
		\left\lVert\lambda\left(\mathcal{I}-\mathcal{P}_{(\mathcal{T}+\mathcal{O}^\bot)}\right)\sum_{i=1}^{\infty}\left(\mathcal{P}_{\mathcal{V}}\mathcal{P}_{(\mathcal{T}+\mathcal{O}^\bot)}\mathcal{P}_{\mathcal{V}}\right)^i\left[\overline{\boldsymbol{\Sigma}}_0\right]\right\lVert\leq\frac{1}{8}
	\end{align}
	with very high probability.
	Consequently, $\left\lVert\mathcal{P}_{\mathcal{T}^\bot}\left[\boldsymbol{\Lambda}_S\right]\right\lVert\leq\frac{1}{4}.$
	
	Next, we deal with $\left\lVert\mathcal{P}_{\mathcal{N}}\left[\boldsymbol{\Lambda}_S\right]\right\lVert_\infty$. Denoting that  $\mathcal{G}\left[\cdot\right]\triangleq\sum_{i=0}^{\infty}\left(\mathcal{P}_{\mathcal{V}}\mathcal{P}_{(\mathcal{T}+\mathcal{O}^\bot)}\mathcal{P}_{\mathcal{V}}\right)^i\mathcal{P}_{\mathcal{V}}\left(\mathcal{I}-\mathcal{P}_{(\mathcal{T}+\mathcal{O}^\bot)}\right)\left[\cdot\right],$ since $\overline{\boldsymbol{\Sigma}}_0\in\mathcal{V}$, we have $\lambda\mathcal{G}^*\left[\overline{\boldsymbol{\Sigma}}_0\right]=\boldsymbol{\Lambda}_S.$ Hence for any $(k,l)\in\mathsf{N}$, 
	\begin{align}
		\left\langle\boldsymbol{E}_{kl},\boldsymbol{\Lambda}_s\right\rangle=\lambda\left\langle\mathcal{G}\left[\boldsymbol{E}_{kl}\right],\overline{\boldsymbol{\Sigma}}_0\right\rangle.
	\end{align}
	Again, invoking Hoeffding's inequality, it is deduced that
	$$Pr\left(\left|\left\langle\mathcal{G}\left[\boldsymbol{E}_{kl}\right],\overline{\boldsymbol{\Sigma}}_0\right\rangle\right|>t\mid \mathsf{V}\right)\leq2\exp\left(-\frac{t^2}{2\left\lVert\mathcal{G}\left[\boldsymbol{E}_{kl}\right]\right\lVert_F^2}\right).$$
	Using the fact that
	$$\mathcal{P}_{\mathcal{V}}\left(\mathcal{I}-\mathcal{P}_{(\mathcal{T}+\mathcal{O}^\bot)}\right)\left[\boldsymbol{E}_{kl}\right]=\mathcal{P}_{\mathcal{V}}\mathcal{P}_{\mathcal{T}}\left(\mathcal{P}_{\mathcal{T}}\mathcal{P}_{\mathcal{O}}\mathcal{P}_{\mathcal{T}}\right)^{-1}\mathcal{P}_{\mathcal{T}}\left[\boldsymbol{E}_{kl}\right],$$
	we have
	\begin{align}
		&\quad\left\lVert\mathcal{G}\left[\boldsymbol{E}_{kl}\right]\right\lVert_F\notag\\&\leq\left\lVert\sum_{i=0}^{\infty}\left(\mathcal{P}_{\mathcal{V}}\mathcal{P}_{(\mathcal{T}+\mathcal{O}^\bot)}\mathcal{P}_{\mathcal{V}}\right)^i\mathcal{P}_{\mathcal{V}}\mathcal{P}_{\mathcal{T}}\left(\mathcal{P}_{\mathcal{T}}\mathcal{P}_{\mathcal{O}}\mathcal{P}_{\mathcal{T}}\right)^{-1}\right\lVert\notag\\&\quad\times\left\lVert\mathcal{P}_{\mathcal{T}}\left[\boldsymbol{E}_{kl}\right]\right\lVert_F\notag\\&\leq\frac{2\sqrt{\rho_s}}{1-2\rho_s}\sqrt{\frac{2\nu r}{n}}
	\end{align}
	provided $\left\lVert\mathcal{P}_{\mathcal{V}}\mathcal{P}_{\mathcal{T}}\right\lVert^2\leq \rho_s$.
	Thereupon, setting $t=\frac{1}{4}$ we're able to conclude that
	\begin{align}
		&\quad Pr\left(\left\lVert\mathcal{P}_{\mathcal{N}}\left[\boldsymbol{\Lambda}_S\right]\right\lVert_\infty>\frac{\lambda}{4}\right)\notag\\&\leq Pr\left(\sup_{\boldsymbol{k},\boldsymbol{l}}\,\left|\left\langle\mathcal{G}\left[\boldsymbol{E}_{kl}\right],\overline{\boldsymbol{\Sigma}}_0\right\rangle\right|>\frac{1}{4}\right)\notag\\&\leq Pr\left(\sup_{\boldsymbol{k},\boldsymbol{l}}\,\left|\left\langle\mathcal{G}\left[\boldsymbol{E}_{kl}\right],\overline{\boldsymbol{\Sigma}}_0\right\rangle\right|>\frac{1}{4}\mid\mathsf{E}\right)+Pr\left(\mathsf{E}^c\right)\notag\\&\leq n^2Pr\left(\left|\left\langle\mathcal{G}\left[\boldsymbol{E}_{kl}\right],\overline{\boldsymbol{\Sigma}}_0\right\rangle\right|>\frac{1}{4}\mid\mathsf{E}\right)+Pr\left(\mathsf{E}^c\right)\notag\\&\leq \frac{2n^2}{\exp\left(\frac{\left(1-2\rho_s\right)^2}{256\rho_s\nu r}n\right)}+Pr\left(\mathsf{E}^c\right).
	\end{align}
	As thus, with a small enough $\rho_s$ or at least a large enough $n$, $\left\lVert\mathcal{P}_{\mathcal{N}}\left[\boldsymbol{\Lambda}_S\right]\right\lVert_\infty\leq\frac{\lambda}{4}$ with very high probability.
	
	The second part of the desired dual certificate is then established.
\end{proof}
\section{Conclusions}
\label{sec:Conclusion}
This paper delves into an yet challenging problem---retrieving the low-rank matrix from its compressed superposition with another sparse matrix. While existing major researches focus on the random-sampling case, we provide the very first theoretical recovery guarantee for the deterministic-sampling case. By appropriately adjusting the construction scheme of the dual certificate, exact identifiability is established with the help of \textit{relative well-conditionedness} and the proposed \textit{restricted approximate $\infty$-isometry property}. Since deterministic sampling is more hardware-friendly, our theory breeds tremendous potential for application in the design of man-made compressive sensing system. 

\bibliographystyle{IEEEtran}
\bibliography{bibfile}

\begin{thebibliography}{10}
\providecommand{\url}[1]{#1}
\csname url@samestyle\endcsname
\providecommand{\newblock}{\relax}
\providecommand{\bibinfo}[2]{#2}
\providecommand{\BIBentrySTDinterwordspacing}{\spaceskip=0pt\relax}
\providecommand{\BIBentryALTinterwordstretchfactor}{4}
\providecommand{\BIBentryALTinterwordspacing}{\spaceskip=\fontdimen2\font plus
\BIBentryALTinterwordstretchfactor\fontdimen3\font minus
  \fontdimen4\font\relax}
\providecommand{\BIBforeignlanguage}[2]{{%
\expandafter\ifx\csname l@#1\endcsname\relax
\typeout{** WARNING: IEEEtran.bst: No hyphenation pattern has been}%
\typeout{** loaded for the language `#1'. Using the pattern for}%
\typeout{** the default language instead.}%
\else
\language=\csname l@#1\endcsname
\fi
#2}}
\providecommand{\BIBdecl}{\relax}
\BIBdecl

\bibitem{bennett2007kdd}
J.~Bennett, C.~Elkan, B.~Liu, P.~Smyth, and D.~Tikk, ``Kdd cup and workshop
  2007,'' \emph{ACM SIGKDD Explorations Newsletter}, vol.~9, no.~2, pp. 51--52,
  2007.

\bibitem{johnson1990matrix}
C.~R. Johnson, ``Matrix completion problems: a survey,'' in \emph{Matrix Theory
  and Applications}, vol.~40, 1990, pp. 171--198.

\bibitem{candes2010power}
E.~J. Cand{\`e}s and T.~Tao, ``The power of convex relaxation: Near-optimal
  matrix completion,'' \emph{IEEE Transactions on Information Theory}, vol.~56,
  no.~5, pp. 2053--2080, 2010.

\bibitem{WOS:000272299900003}
E.~J. Candes and B.~Recht, ``Exact matrix completion via convex optimization,''
  \emph{FOUNDATIONS OF COMPUTATIONAL MATHEMATICS}, vol.~9, no.~6, pp. 717--772,
  DEC 2009.

\bibitem{keshavan2010matrix}
R.~H. Keshavan, A.~Montanari, and S.~Oh, ``Matrix completion from a few
  entries,'' \emph{IEEE Transactions on Information Theory}, vol.~56, no.~6,
  pp. 2980--2998, 2010.

\bibitem{recht2011simpler}
B.~Recht, ``A simpler approach to matrix completion.'' \emph{Journal of Machine
  Learning Research}, vol.~12, no.~12, 2011.

\bibitem{7064749}
Y.~Chen, ``Incoherence-optimal matrix completion,'' \emph{IEEE Transactions on
  Information Theory}, vol.~61, no.~5, pp. 2909--2923, 2015.

\bibitem{doi:10.1137/080738970}
\BIBentryALTinterwordspacing
J.-F. Cai, E.~J. Cand\`{e}s, and Z.~Shen, ``A singular value thresholding
  algorithm for matrix completion,'' \emph{SIAM Journal on Optimization},
  vol.~20, no.~4, pp. 1956--1982, 2010. [Online]. Available:
  \url{https://doi.org/10.1137/080738970}
\BIBentrySTDinterwordspacing

\bibitem{6262492}
G.~Marjanovic and V.~Solo, ``On $l_q$ optimization and matrix completion,''
  \emph{IEEE Transactions on Signal Processing}, vol.~60, no.~11, pp.
  5714--5724, 2012.

\bibitem{6979050}
M.~Hardt, ``Understanding alternating minimization for matrix completion,'' in
  \emph{2014 IEEE 55th Annual Symposium on Foundations of Computer Science},
  2014, pp. 651--660.

\bibitem{8400447}
A.~Ramlatchan, M.~Yang, Q.~Liu, M.~Li, J.~Wang, and Y.~Li, ``A survey of matrix
  completion methods for recommendation systems,'' \emph{Big Data Mining and
  Analytics}, vol.~1, no.~4, pp. 308--323, 2018.

\bibitem{9363502}
X.~Yuan, D.~J. Brady, and A.~K. Katsaggelos, ``Snapshot compressive imaging:
  Theory, algorithms, and applications,'' \emph{IEEE Signal Processing
  Magazine}, vol.~38, no.~2, pp. 65--88, 2021.

\bibitem{chen2014coherent}
Y.~Chen, S.~Bhojanapalli, S.~Sanghavi, and R.~Ward, ``Coherent matrix
  completion,'' in \emph{International Conference on Machine Learning}.\hskip
  1em plus 0.5em minus 0.4em\relax PMLR, 2014, pp. 674--682.

\bibitem{bhattacharya2022matrix}
S.~Bhattacharya and S.~Chatterjee, ``Matrix completion with data-dependent
  missingness probabilities,'' \emph{IEEE Transactions on Information Theory},
  vol.~68, no.~10, pp. 6762--6773, 2022.

\bibitem{pimentel2016characterization}
D.~L. Pimentel-Alarc{\'o}n, N.~Boston, and R.~D. Nowak, ``A characterization of
  deterministic sampling patterns for low-rank matrix completion,'' \emph{IEEE
  Journal of Selected Topics in Signal Processing}, vol.~10, no.~4, pp.
  623--636, 2016.

\bibitem{shapiro2018matrix}
A.~Shapiro, Y.~Xie, and R.~Zhang, ``Matrix completion with deterministic
  pattern: A geometric perspective,'' \emph{IEEE Transactions on Signal
  Processing}, vol.~67, no.~4, pp. 1088--1103, 2018.

\bibitem{chatterjee2020deterministic}
S.~Chatterjee, ``A deterministic theory of low rank matrix completion,''
  \emph{IEEE Transactions on Information Theory}, vol.~66, no.~12, pp.
  8046--8055, 2020.

\bibitem{burnwal2020deterministic}
S.~P. Burnwal and M.~Vidyasagar, ``Deterministic completion of rectangular
  matrices using asymmetric ramanujan graphs: Exact and stable recovery,''
  \emph{IEEE Transactions on Signal Processing}, vol.~68, pp. 3834--3848, 2020.

\bibitem{liu2019matrix}
G.~Liu, Q.~Liu, X.-T. Yuan, and M.~Wang, ``Matrix completion with deterministic
  sampling: Theories and methods,'' \emph{IEEE Transactions on Pattern Analysis
  and Machine intelligence}, vol.~43, no.~2, pp. 549--566, 2019.

\bibitem{liu2017new}
G.~Liu, Q.~Liu, and X.~Yuan, ``A new theory for matrix completion,''
  \emph{Advances in Neural Information Processing Systems}, vol.~30, 2017.

\bibitem{9851461}
G.~Liu and W.~Zhang, ``Recovery of future data via convolution nuclear norm
  minimization,'' \emph{IEEE Transactions on Information Theory}, vol.~69,
  no.~1, pp. 650--665, 2023.

\bibitem{10.1109/TIT.2011.2173156}
\BIBentryALTinterwordspacing
H.~Xu, C.~Caramanis, and S.~Sanghavi, ``Robust pca via outlier pursuit,''
  \emph{IEEE Trans. Inf. Theor.}, vol.~58, no.~5, p. 3047–3064, may 2012.
  [Online]. Available: \url{https://doi.org/10.1109/TIT.2011.2173156}
\BIBentrySTDinterwordspacing

\bibitem{Wright_Ma_2022}
J.~Wright and Y.~Ma, \emph{High-Dimensional Data Analysis with Low-Dimensional
  Models: Principles, Computation, and Applications}.\hskip 1em plus 0.5em
  minus 0.4em\relax Cambridge University Press, 2022.

\bibitem{candes2010matrix}
E.~J. Candes and Y.~Plan, ``Matrix completion with noise,'' \emph{Proceedings
  of the IEEE}, vol.~98, no.~6, pp. 925--936, 2010.

\bibitem{keshavan2009matrix}
R.~Keshavan, A.~Montanari, and S.~Oh, ``Matrix completion from noisy entries,''
  \emph{Advances in neural information processing systems}, vol.~22, 2009.

\bibitem{candes2011robust}
E.~J. Cand{\`e}s, X.~Li, Y.~Ma, and J.~Wright, ``Robust principal component
  analysis?'' \emph{Journal of the ACM (JACM)}, vol.~58, no.~3, pp. 1--37,
  2011.

\bibitem{li2013compressed}
X.~Li, ``Compressed sensing and matrix completion with constant proportion of
  corruptions,'' \emph{Constructive Approximation}, vol.~37, pp. 73--99, 2013.

\bibitem{cherapanamjeri2017nearly}
Y.~Cherapanamjeri, K.~Gupta, and P.~Jain, ``Nearly optimal robust matrix
  completion,'' in \emph{International Conference on Machine Learning}.\hskip
  1em plus 0.5em minus 0.4em\relax PMLR, 2017, pp. 797--805.

\bibitem{klopp2017robust}
O.~Klopp, K.~Lounici, and A.~B. Tsybakov, ``Robust matrix completion,''
  \emph{Probability Theory and Related Fields}, vol. 169, pp. 523--564, 2017.

\bibitem{ashraphijuo2017deterministic}
M.~Ashraphijuo, V.~Aggarwal, and X.~Wang, ``On deterministic sampling patterns
  for robust low-rank matrix completion,'' \emph{IEEE Signal Processing
  Letters}, vol.~25, no.~3, pp. 343--347, 2017.

\bibitem{gross2011recovering}
D.~Gross, ``Recovering low-rank matrices from few coefficients in any basis,''
  \emph{IEEE Transactions on Information Theory}, vol.~57, no.~3, pp.
  1548--1566, 2011.

\bibitem{vershynin2010introduction}
R.~Vershynin, ``Introduction to the non-asymptotic analysis of random
  matrices,'' \emph{arXiv preprint arXiv:1011.3027}, 2010.

\bibitem{ledoux2001concentration}
M.~Ledoux, \emph{The concentration of measure phenomenon}.\hskip 1em plus 0.5em
  minus 0.4em\relax American Mathematical Soc., 2001, no.~89.

\end{thebibliography}

\end{document}